%% file: ICviaConvexOptimization_Arxiv.tex
\documentclass[conference]{IEEEtran}
\usepackage{times}
\usepackage{epsfig}
\usepackage{cite}
\usepackage{amsmath}
\DeclareMathOperator{\Tr}{Tr}
\usepackage{amsfonts}
\usepackage{mathtools}
\usepackage{graphicx}
\usepackage{amssymb}
\usepackage{amstext}
\usepackage{subcaption}
\usepackage{latexsym}
\usepackage{color}
\usepackage[ruled,linesnumbered]{algorithm2e}
\usepackage{ifthen}
\usepackage{multirow}
\usepackage{hyperref}
\usepackage{verbatim}
\usepackage{array,tabularx}
\usepackage[mathscr]{euscript}
\usepackage{tikz}
\usetikzlibrary{arrows,positioning,automata,calc,shapes}
\usepackage{pgfplots}
\pgfplotsset{compat=newest} 
\pgfplotsset{plot coordinates/math parser=false}
\newlength\figureheight 
 \newlength\figurewidth 
\newcommand{\be}{\begin{equation}}
\newcommand{\ee}{\end{equation}}
\newcommand{\bear}{\begin{eqnarray}}
\newcommand{\eear}{\end{eqnarray}}
\newcommand{\bears}{\begin{eqnarray*}}
\newcommand{\eears}{\end{eqnarray*}}
\newcommand{\bi}{\begin{itemize}}
\newcommand{\ei}{\end{itemize}}
\newcommand{\ben}{\begin{enumerate}}
\newcommand{\een}{\end{enumerate}}

\newtheorem{theorem}{Theorem}%[section]
\newtheorem{lemma}[theorem]{Lemma}

\newtheorem{example}[theorem]{Example}%[section]
\newcommand{\rank}{\text{rank}}

\DeclareMathOperator{\minrk}{minrk}

\IEEEoverridecommandlockouts

\renewcommand{\baselinestretch}{1}

\begin{document}
%\title{Index Coding Via Convex Optimization}

\title{ Index Coding and Network Coding via Rank Minimization}
%\title{Secure MDS codes with Exact Repair for Distributed Storage Systems}

%\title{Secure MDS codes for Distributed Storage Systems with Exact Repair}
\author{% aefraer \IEEEauthorrefmark{baer}
\IEEEauthorblockN{Xiao Huang and Salim El Rouayheb\\ ECE Department, IIT, Chicago\\ Emails: xhuang31@hawk.iit.edu, salim@iit.edu}
}
\maketitle
\begin{abstract}

Index codes reduce the number of bits broadcast by a wireless transmitter to a number of receivers with different demands and with side information. It is known that the problem of finding  optimal linear index codes is NP-hard. We  investigate the performance of different heuristics based on rank minimization and matrix completion methods, such as alternating projections and alternating minimization,  for constructing linear index codes over the reals.  As a summary of our results,  the alternating projections method gives the best results in terms of   minimizing the number of broadcast bits and convergence rate  and leads to up to $13\%$  savings in average communication cost compared to graph coloring algorithms studied in the literature. Moreover, we describe how the proposed methods can be used to construct linear network codes for non-multicast networks. Our computer code is available online.
\end{abstract}

\section{Introduction}\label{sec:Intro}

 %%%%%%%%%%%%%Figure1%%%%%%%%%%%%%%%%%%%%
\begin{figure}[b]
\centering
\includegraphics[scale=0.5]{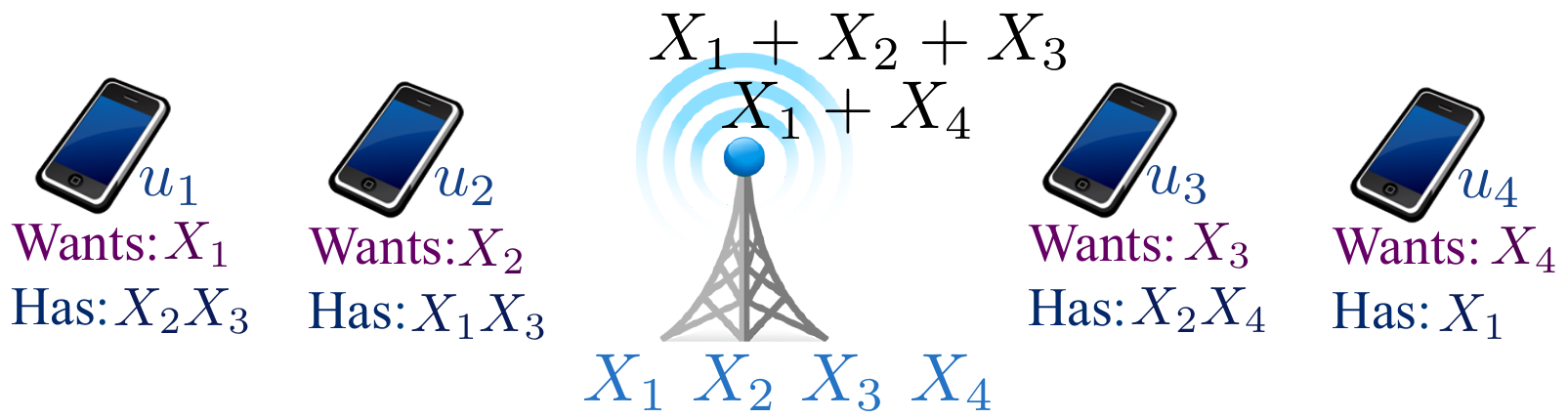}\vspace{-.2cm}  
\caption{ An index code example.}
\label{fig:ex1}
\end{figure}
% {\red  1) we need to mention somewhere that these are heuristics. 2) We need to decide on the figure whether we want undirected or directed. 3) We need to mention that the optimization is over the reals earlier and maybe justify it more. 4) We need to rewrite the abstract.}

We investigate the performance of different rank minimization heuristics for constructing linear index codes \cite{BirkKol,ISCOD2006}, and therefore linear network codes using the equivalence in \cite{RSG10, EEL13}. Index codes  reduce the number of bits broadcast by a wireless transmitter that wishes to satisfy the different demands of a  number of receivers with side information in their caches. Fig.~\ref{fig:ex1} illustrates an  index coding example.   A wireless transmitter  has $n=4$  packets, or messages,  $X_1,\dots, X_4,$  and there are  $n=4$ users (receivers) $u_1,\dots, u_4$. User $u_i$ wants packet $X_i$ and  has a subset of the packets as side information. The packets in the cache could have been obtained in a number of ways:  packets downloaded earlier,  overheard packets or packets downloaded during off-peak network hours.   Each user reports to the transmitter the indices of its requested  and cached packets, hence the nomenclature index coding \cite{BBJK06}.   Assuming  an error-free  broadcast channel, the objective is to design a coding  scheme at the transmitter, called index code,  that  satisfies the demands of all the users while minimizing the number of broadcast messages. For instance,  the transmitter can always satisfy the demands of all the users by broadcasting all the four packets. However, it can save half of the broadcast rate by transmitting only  $2$ coded packets, $X_1+X_2+X_3$ and $X_1+X_4$ to the users.  Each user can decode its requested packet by using the broadcast packets and its side information. The problem that we focus on here is  is how to construct linear index codes that minimize the number of broadcast messages.

 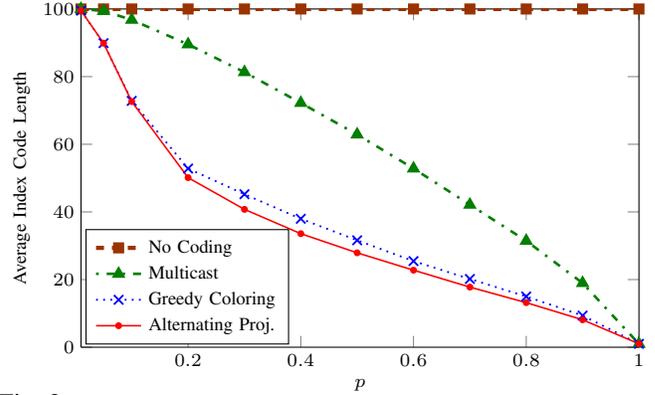
\begin{figure}[t]
\centering  \scriptsize
 \setlength\figureheight{4.5cm} 
\setlength\figurewidth{7.8cm}
      \input{directed_AP_n=100_Rmin.tikz} \vspace{-.3cm}      
\caption{\footnotesize Comparison of different methods for constructing scalar linear index codes for $n=100$ users  and messages. Each user caches each message independently with probability $p$ (except its requested message).}
\label{fig:simu1}  
\vspace{-.6 cm} 
\end{figure}

\noindent\paragraph*{Contribution} Answering the  question above turns out to be an NP-hard problem in general  \cite{RCA07, Pee96,LangbergSprintson11}. Motivated by a connection between  linear index codes and rank minimization \cite{BBJK06} (details in Sec.~\ref{sec:minrk}), we  propose to use   rank minimization and matrix completion methods to   construct linear index codes. The underlying matrices representing an index coding problem have a special structure that affects the performance of these methods. For instance, the celebrated nuclear norm minimization method \cite{candes2009exact, recht2010guaranteed} does not perform well here. We present our findings on the performance of different other methods, such as alternating projections, directional alternating projections and alternating minimization, through extensive simulation results on random instances of the index coding problem. These methods are  performed over the real numbers and give linear index codes over the reals which have applications to topological interference management in wireless networks \cite{Syed, jafar2013topological}.
As a sample of our results, Fig.~\ref{fig:simu1} compares the performance of index codes obtained by the Alternating Projection (AP) method to other methods studied in the literature. We assumed that packets are cached independently and randomly with probability $p$. The figure shows the savings in communication cost resulting from using index codes compared to no-coding and  multicast network coding (all users decode all messages). The AP method leads to up to $13\%$ average savings in broadcast messages compared to graph coloring  \cite{BirkKol, RCA07}.

 Over the recent years, several  connections have been established between index coding and other problems. These connections can be leveraged to apply the rank minimization methods presented here  to these equivalent problems.
For instance, using the reduction  between index coding and network coding devised in \cite{RSG10, EEL13} to show the equivalence of the two problems,   the methods proposed here could be readily applied to construct linear network codes over the reals  \cite{Jaggi08, Gastpar12} for general non-multicast networks. Similarly, these methods can be used to construct  certain class of locally repairable codes (over the reals) using the duality between index codes and locally repairable codes established in \cite{Arya, Dimakis14}. Our  computer code for constructing linear index codes, network codes and locally repairable codes is available online \cite{APlink}.

\noindent\paragraph*{Related work} Index coding was  introduced  by Birk and Kol in  \cite{BirkKol} as a caching problem in satellite communications. The work of  \cite{BBJK06} established the connection between linear index codes and  the minimum rank of the side information graph representing the problem.
% In \cite{RSG08, RSG09,RSG10},  an equivalence was established between linear network codes, linear index codes and matroid linear representations. It followed from that equivalence that  linear codes are not optimal. This equivalence was later strengthened to non-linear codes in \cite{EEL13}.
The sub-optimality of linear index codes was  shown in \cite{LS07, RSG08, BKL11}.  The work of   \cite{ALSW08} further explored the connection to graph coloring  and  studied  properties of index coding on  the direct sums of graphs.  Linear programming bounds were  studied in \cite{blasiak2010index} and connections to local graph coloring   and multiple  unicast networks were investigated in \cite{DimakisLocal13} and \cite{Dimakis14}, respectively. The work in  \cite{haviv2012linear} investigated the property of  index codes on random graphs. Tools from network information theory  \cite{KimIndex13,   Kim14} and distributed source coding \cite{unal2013general} were also used to  tackle the index coding problem.  Related to index coding is the line of work on  distributed caching in  \cite{MohammadCaching,maddah2013decentralized}.  Recently, a matrix completion method for constructing linear index codes over finite fields was proposed  in \cite{Hassibi2015}, and a method for constructing quasi-linear vector    network codes over the reals was described in \cite{Schwartz14}.

\noindent\paragraph*{Organization} The rest of the paper is organized as follows. In Section~\ref{sec:Model}, we describe the mathematical model of the index coding problem and the assumptions we make. In Section~\ref{sec:Conn}, we summarize the connections of index coding to graph coloring, rank minimization and topological interference management. In Section~\ref{sec:undir}, we focus on index coding instances that can be represented by undirected graphs. We describe the different rank minimization methods and  our  simulation results. In Section~\ref{sec:dir}, we describe the performance of these methods for directed graphs. In Section\ref{sec:RnkMin}, we elaborate more on the use of rank minimization methods for constructing linear network codes. We conclude in Section~\ref{sec:conclusion}.

\section{Model}\label{sec:Model}
An instance of the index coding problem is defined as follows. A transmitter or server holds a set of $n$  messages or packets, $\mathcal{X}=\{X_1,\dots,X_n\}$, where the $X_i$'s belong to some alphabet. There are $m$ users,  $u_1, \dots,u_m$. Let $W_i\subset \mathcal{X}$ (``wants" set) represents the packets requested by $u_i$, and the set $H_i\subset\mathcal{X}$ (``has" set) represents the packets available to $u_i$ as side information in its cache. WLOG, we can assume that $W_i$ contains only  one packet, otherwise the user can be represented by  multiple users satisfying this condition.  We assume that  initially the transmitter does not know which packets are cached at each user, and the users tell the transmitter the indices of the packets they have in an initial stage. Typically, the alphabet size (packet length) is much larger than the number of packets $n$, so the overhead in the initial stage is negligible. The transmitter uses an error-free broadcast channel to transmit information to the terminals. The objective is to design a coding  scheme at the transmitter, called index code,  that  satisfies the demands of all the users while minimizing the number of broadcast bits. 
We will focus on linear index codes in which the messages belong to a certain field ($GF(q)$ or $\mathbb{R}$) and the transmitted messages are linear combinations of these messages. Linear index codes are known not to be optimal \cite{RSG08} and the gap to optimality can be arbitrarily large \cite{LS07}.  However, we focus  on linear codes due to their tractability. For clarity of exposition, we make the following two assumptions:

\begin{enumerate}
\item  The number of users is equal to the number of messages ($n=m$). We will assume that  user $u_i$ requests message $X_i$, i.e., $W_i=\{X_i\}$. It was shown in \cite{jafar2012elements} that any general instance, $m\geq n$, can be reduced to this model with no loss of generality for  linear codes.

\item The messages are atomic units that cannot be divided. This corresponds to scalar linear index codes. We refer to the number of broadcast messages as the index code length. We denote by $L_{min}$ the minimum number of broadcast messages achieved by scalar linear index codes. Our methods could be easily extended to vector linear index codes for a  givenblock length.
\end{enumerate}

\section{Connections to Other Problems}\label{sec:Conn}
\subsection{Index Coding \& Graph Coloring}
 %%%%%%%%%%%%%Figure3%%%%%%%%%%%%%%%%%%%% 
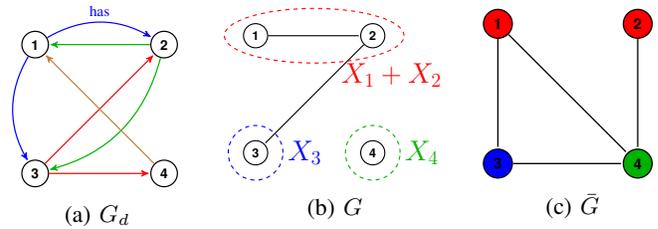
\begin{figure}[t]
\centering  
\begin{subfigure}{0.14\textwidth}
  \resizebox{0.95\textwidth}{!}{
      \begin{tikzpicture}[->,>=stealth',shorten >=1pt,auto,node distance=3cm,
  thick]
  \tikzstyle{sender} = [circle,draw,font=\sffamily\small \bfseries]
  \node[sender] (1) {1};
  \node[sender] (2) [right of =1] {2};
  \node[sender] (3) [below of =1] {3};
  \node[sender] (4) [below of =2] {4};
   \path[every node/.style={}]
   (1)   edge  [bend left, color=blue] node {has} (2)
           edge [bend right, color=blue] (3) 
   (2)   edge [ color=green!70!black] (1) 
           edge [bend left , color=green!70!black] (3)
   (3)   edge [ color=red] (2) 
           edge [ color=red] (4)
   (4)   edge [ color=brown] (1) ; 
\end{tikzpicture}
}
\caption{$G_d$}
\end{subfigure}
~
\begin{subfigure}{0.17\textwidth}
  \resizebox{1\textwidth}{!}{
     \begin{tikzpicture}[-,>=stealth',shorten >=1pt,auto,node distance=3cm,
  thick]
  \tikzstyle{sender} = [circle,draw,font=\sffamily\small \bfseries]
  \node[sender] (1) {1};
  \node[sender] (2) [right of =1] {2};
  \node[sender] (3) [below of =1] {3};
  \node[sender] (4) [below of =2] {4};
  \path[every node/.style={}]
   (1)   edge (2)  
   (2)   edge (3); 
 \draw[red,thick,dashed] ($(1)+(1.5, 0)$) ellipse (2.3cm and 0.7cm) node [below right =0.8cm, color=red, font=\huge] {$X_1+X_2$};
  \draw[blue,thick,dashed] ($(3)$) circle (0.7cm) node [right =0.7cm, color=blue, font=\huge] {$X_3$} ;
  \draw[green!70!black,thick,dashed] ($(4)$) circle (0.7cm)node [right =0.7cm, color=green!70!black, font=\huge] {$X_4$}; 
\end{tikzpicture}
}
\caption{$G$}
\end{subfigure}
~
\begin{subfigure}{0.14\textwidth}
  \resizebox{0.95\textwidth}{!}{
     \begin{tikzpicture}[-,>=stealth',shorten >=1pt,auto,node distance=3cm,
  thick]
  \tikzstyle{sender1} = [circle,fill=red,draw,font=\sffamily\small \bfseries]
    \tikzstyle{sender3} = [circle,fill=blue,draw,font=\sffamily\small \bfseries]
    \tikzstyle{sender4} = [circle,fill=green!70!black,draw,font=\sffamily\small \bfseries]
  \node[sender1] (1) {1};
  \node[sender1] (2) [right of =1] {2};
  \node[sender3] (3) [below of =1] {3};
  \node[sender4] (4) [below of =2] {4};
  \path[every node/.style={}]
   (1)   edge (3)  
           edge (4)  
   (2)   edge (4)
   (3)   edge (4); 
\end{tikzpicture}
}
\caption{$\bar{G}$}
\end{subfigure} 
\vspace{-.1cm}  
\caption{\footnotesize (a) Side information graph $G_d$ of the example in Fig.~\ref{fig:ex1}. (b) A clique cover for its undirected subgraph $G$ and the corresponding index code. (c) Graph coloring of the complement graph $\bar{G}$ corresponding to the clique cover in $G$.} 
\label{fig:graph}
\vspace{-.5cm}  
\end{figure}
% Recently, there has been lots of interest in using graph theoretic tools to study  index coding  \cite{ALSW08, blasiak2010index, Kartik, BirkKol, RCA07, SV10}. Achievable schemes can be devised by  partitioning the  graph into  subgraphs that are easier to handle. The high-level question that we seek to answer here is:  How far can these  graph theoretic schemes be from the information theoretic limits? Not far. This is the ``one-bit" answer that our initial investigations seem to point to.
%

The minimum scalar linear index codes length $L_{min}$ can be upper bounded by the chromatic number of a certain graph. An index coding problem, with $n$ messages and $m=n$ users\footnote{In the case where there are more users than messages, i.e., $m>n$, the index coding problem can be represented by a multigraph \cite{ALSW08} or a bipartite graph \cite{Dimakis14}.}, can be represented by a directed graph $G_d$, referred to as side information graph,  defined on the vertex set $\{1,2,...,n\}$. An edge $(i,j)$ is in the edge set of $G_d$ iff  user $u_i$ has packet $X_j$ as side information. 

%%\begin{wrapfigure}{r}{0.3\textwidth}
%%  \begin{center}
%%    \includegraphics[trim=0cm 0cm 0cm .5cm, width=0.3\textwidth]{Figures/Graph2.pdf}
%%  \end{center}\vspace{-.3cm}
%%  \caption{{\footnotesize Clique cover.}}\label{fig:Mmatrix}\label{fig:graph}\vspace{-.7cm}
%%\end{wrapfigure}
%
%
The side information graph $G_d$ representing the instance in  Fig.~\ref{fig:ex1} is depicted in Fig.~\ref{fig:graph}(a). Its maximal undirected subgraph $G$ in Fig.~\ref{fig:graph}(b) is obtained from $G_d$ by replacing any two  edges in opposite directions  by an undirected edge, and removing the remaining directed  edges. We will say that $G_d$ is undirected if $G_d$ and $G$ are the same graph. A fully connected subgraph (clique) of $G$ represents a subset of users that can be satisfied simultaneously by broadcasting a single coded packet that is the XOR of all the packets indexed by  the clique.  Therefore, a partition of $G$ into cliques gives a scalar linear index code over any field. We can optimize such a partition in order to obtain a  minimum number  of cliques. Such a number is called the minimum clique cover,  $\bar{\chi}(G)$, of $G$ (see Fig.~\ref{fig:graph}(b)). Note that the minimum clique cover number $\bar{\chi}(G)$ is equal to the chromatic number $\chi(\bar{G})$ of  the complement graph $\bar{G}$ (Fig.~\ref{fig:graph}(c)), and therefore finding it is an NP-hard problem \cite{Karp,garey2002computers}.  Fig.~\ref{fig:graph}(b) shows  a minimum clique cover  of $G$ and the resulting index code of rate $3>2$, and therefore clique cover based index codes are not  necessarily optimal. 
 Nevertheless, it  is the basis of many greedy heuristics   in the literature \cite{ISCOD2006, RCA07}. 
 
 A lower bound on $L_{min}$ is the independence number $\alpha(G)$ which is  the maximum number of vertices with no edge between any two of them. To see this, consider the sub-problem formed by the users corresponding to an independent set of $G$ and their messages. In this sub-problem, users do not have any side information and therefore all the messages must be transmitted. We summarize the results above in the following Lemma. %\cite{BirkKol, BBJK06, RCA07}.
 
 \begin{lemma}
 $\alpha(G)\leq L_{min}\leq \chi(\bar{G})$.
 \end{lemma}

\subsection{Index Coding \& Rank Minimization}\label{sec:minrk}

It was shown in \cite{BBJK06} that finding an optimal scalar linear index code is equivalent to minimizing the rank of a  certain matrix $M$. For instance, this matrix $M$ for the example in Fig.~\ref{fig:ex1} is given by 
\begin{equation*}
M=\begin{tabular}{m{3pt} m{8pt} m{5pt} m{3pt}m{3pt}}
 & $X_1$ & $X_2$ & $X_3$ & $X_4$ \\         
 $u_1 $ & \multicolumn{4}{c}{\multirow{4}{*}{$\begin{pmatrix}
 1 & * & * & 0\\ 
 * & 1 & * & 0\\ 
 0 & * & 1 & *\\ 
 * & 0 & 0 & 1
\end{pmatrix}$}} \\
$u_2 $& \\
$u_3 $&  \\
$u_4$ & \\
\end{tabular}.
\end{equation*}

The matrix $M$ is constructed by setting all the diagonal elements to $1$'s,  a star  in the $(i,j)^{th}$ position if edge $(i,j)$ exists in $G_d$, i.e., user $u_i$ caches packet $X_j$,  otherwise the entry is $0$. 
 The intuition is that the $i$th row of $M$ represents the linear coefficients of the coded packet that user $u_i$ will use to decode $X_i$. Hence,  the zero entries  enforce that this coded packet does not involve packets that   $u_i$ does not have as side information.  The packets that $u_i$ has as side information can always be subtracted out of the linear combination.
 The goal is to  choose values for the stars ``$*$" from a certain field $\mathbb{F}$ such that  the  rank of $M$ is minimized. 
 The saving in transmitted messages  can be  achieved by making the transmitter only broadcast  the coded packets that generate the row space of $M$.  It turns out that this formulation of index coding coincides with the minimum rank  of a graph $G$, $\minrk(G)$, defined by Haemers  \cite{shannon1956zero}.  Therefore, the optimal rate for a scalar linear index code is $L_{min}=\minrk(G)\leq \bar{\chi}(G_d)$\cite{BBJK06}.

 %%%%%%%%%%%%%Figure4%%%%%%%%%%%%%%%%%%%% 
 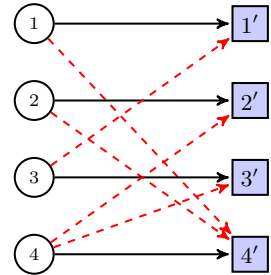
\begin{figure}[t]
 \begin{center}
\resizebox{0.2\textwidth}{!}{
     \begin{tikzpicture}[->,>=stealth',shorten >=1pt,auto,node distance=3cm,
  thick]
  \tikzstyle{sender} = [circle,draw,font=\sffamily\scriptsize \bfseries]
  \tikzstyle{receiver} = [rectangle,draw,fill=blue!20, font=\sffamily \small \bfseries]
  \node[sender] (1) {$1$};
  \node[receiver] (2) [right of  =1] {$1'$};
  \node[sender] (3) [below =0.5cm of 1] {$2$};
  \node[receiver] (4) [right  of =3] {$2'$};
  \node[sender] (5) [below =0.5cm of 3] {$3$};
  \node[receiver] (6) [right  of =5] {$3'$};
  \node[sender] (7) [below =0.5cm of 5] {$4$};
  \node[receiver] (8) [right  of =7] {$4'$};
   \path[every node/.style={}]
   (1)   edge (2) 
           edge [dashed, color=red]  (8) 
   (3)   edge  (4)
           edge [dashed, color=red] (8) 
   (5)   edge (6) 
           edge [dashed, color=red]  (2)
   (7)   edge (8) 
           edge [dashed, color=red] (4)
           edge [dashed, color=red] (6);
\end{tikzpicture}
}
 \vspace{-0.2cm}
\caption{ \footnotesize Interference management problem equivalent to the index coding instance in Fig.~\ref{fig:ex1} for the linear case. Circles represent transmitter nodes connected by black links to their intended receivers represented by squares. Dashed red links represent the interference between the different transmitters and receivers. }\label{fig:IM}
\end{center}
\vspace{-0.6cm}
\end{figure}
 
\subsection{Index Coding \& Topological Interference Management}\label{subsec:Interference}

It was shown in \cite{Syed} that, in the linear case, the  index coding problem is equivalent to the  topological interference management  problem in wireless networks.  The latter problem consists of  finding optimal  transmission schemes in interference networks with no channel state information at the transmitter.  This equivalence holds over any field, in particular the field of real numbers $\mathbb{R}$ on which we focus in this paper. 

We will briefly describe this equivalence using an example and  refer the interested reader to the results in \cite{Syed} and related literature \cite{jafar2013topological, sun2013index} for more details.  Fig.~\ref{fig:IM} depicts the wireless interference network that is equivalent  to the index coding problem in Fig.~\ref{fig:ex1}. Black solid links connect a transmitter $i$ with its intended receiver $i'$. Dashed red links indicate the set of receivers with which a transmitter node interferes when transmitting. For example, transmitter $1$ interferes with receiver $4'$. Transmitted signals are added ``in the air" and a receiver will receive the sum of his intended signal plus the interference on the red links. The intuition behind the connection to index coding can be explained as follows. Take for example receiver $1'$, it does not suffer of interference from $2$ and $4$ which is equivalent to  $1'$ getting interference from all the transmitters and $1'$ possessing the messages of $2$ and $4$ as side information, so it can cancel them out. One communication scheme for this network consists of letting each transmitter sends his message during a different time slot while the other transmitters are ``off". This corresponds to the trivial index coding solution of sending all the messages. The index code in Fig.~\ref{fig:ex1} gives a more efficient scheme for interference network in Fig.~\ref{fig:IM}:  nodes $\{1, 2, 3 \}$ transmit together in the first time slot, then   nodes $\{1, 4\}$ transmit together in the second time slot.

 %%%%%%%%%%%%%Figure5%%%%%%%%%%%%%%%%%%%% 
 \begin{figure}[t]
 \begin{center}
\begin{subfigure}{0.23\textwidth}
  \resizebox{1\textwidth}{!}{
      \begin{tikzpicture}
      \draw(2, 0) arc (10:38:10cm) ;
      \draw(4.8, 4.8) arc (120:176:6cm);
      \filldraw[black] (4,4.247) circle (1pt) node[anchor=west] {$d_1$};
      \draw[dashed](4,4.247) --(1.078,2.771);
      \filldraw[black] (1.078,2.771) circle (1pt) node[anchor=east] {$c_1$} node [below left  = 0.7 cm] {$\mathcal{C}$};
      \draw[dashed](1.078,2.771)--(2.372, 2.161);
      \filldraw[black] (2.372, 2.161) circle (1pt) node[anchor=west] {$d_2$} node [right  = 0.7 cm] {$\mathcal{D}$};
      \draw[dashed](2.372, 2.161)--(1.496, 1.827);
      \filldraw[black] (1.496, 1.827) circle (1pt) node[anchor=east] {$c_2$};
      \draw[dashed](1.496, 1.827) --(2.142, 1.599);
      \filldraw[black] (2.142, 1.599)circle (1pt) node[anchor=west] {$d_3$};
      \draw[dashed](2.142, 1.599)--(1.637, 1.43);
      \filldraw[black] (1.637, 1.43) circle (1pt) node[anchor=east] {$c_3$};
      \filldraw[black] (1.88,0.58)circle (1pt) node[anchor=west] {$d^*$};
\end{tikzpicture}
}
\vspace{-0.5cm}\caption{}
\end{subfigure}
\begin{subfigure}{0.23\textwidth}
  \resizebox{1\textwidth}{!}{
\begin{tikzpicture}
\fill[gray!40!white] (4.5, 0)--(2.8, 0)   -- (3.5,5) -- (5.2, 5);
\draw (4.5, 0)--(2.8, 0)   -- (3.5,5) -- (5.2, 5) ;
\draw (0, 4.5)-- (1, 4.46) arc (86:-45:1.2 cm).. controls (0.8, 1.2) and (5, 0.8) .. (2, 0) ;
 \filldraw[black] (3.36,4) circle (1pt) node[anchor=west] {$d_1$};
\draw [->, >=stealth]   (3.4 ,4)--(2.05, 3.65) node [ left ]{ {\bf Rank} $M \leq r$};
\draw [->, >=stealth]  (2.05, 3.65) -- (3.28, 3.46) ;
\draw [->, >=stealth]  (3.28, 3.46) -- (2.11,3.39) node [below left  = 0.7 and 1.3 cm] {$\mathcal{C}$};
\draw [->, >=stealth]  (2.11,3.39)  -- (3.25, 3.21)  node [below right  =0.7 cm] {$\mathcal{D}$};
\draw [->, >=stealth]  (3.25, 3.21) -- (2.11, 3.22) ;
\draw [->, >=stealth]  (2.11, 3.22) -- (3.23, 3.07);
\draw [->, >=stealth]  (3.23, 3.07) -- (2.11,3.14);
\draw [->, >=stealth]  (2.11,3.14) -- (3.22, 3) ;
\draw [->, >=stealth]  (3.22, 3) -- (2.11,3.15) ;
\end{tikzpicture}
}
\caption{}
\end{subfigure} 
\vspace{-0.3cm}
\caption{\footnotesize (a) Alternating Projections (AP) method between two convex sets. (b) AP method for the index coding problem (see Eqs.~\eqref{eq:regionC} and \eqref{eq:regionD}).}
\label{fig:AP}
\end{center}
\vspace{-.6cm}
\end{figure}
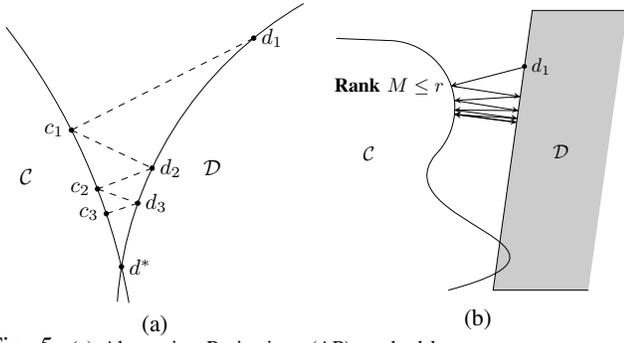

\section{Index Coding on Random Undirected Graphs}\label{sec:undir}

We start by considering undirected  side information graphs $G_d$  ($G_d=G$), i.e., for every directed edge $(i,j)$ in $G$ there is an edge $(j,i)$ in the opposite direction. 
% Index codes on undirected graphs are  more studied in the literature than those on directed graphs. Moreover, as will see later, our results indicate that even if the graph $G_d$ is directed, applying the optimization methods on the underlying undirected graph $G$ may be quite as good in general.
Our approach is to use convex optimizing methods to find $L_{min}$ by minimizing the rank of the matrix $M$ over the reals.  The problem of rank minimization has been extensively studied in system theory \cite{fazel2001rank,  fazel2004rank}.  In \cite{recht2010guaranteed, candes2009exact}, it was shown that the convex relaxation that replaces the rank function by the nuclear norm (sum of singular values)  leads to finding the minimum rank with high probability under certain conditions on the matrix rank and the number of fixed  (observed) entries in the matrix. However, these results do not carry over directly to the index coding problem because the model there assumes the location of the fixed entries is chosen uniformly at random. In contrast, the  index coding matrix $M$ has a specific structure that dictates all the  the diagonal entries to be equal to one. Indeed,  the semi-definite program (SDP) relaxation in \cite{candes2009exact} always output the  maximum rank $n$ (instead of the minimum rank) which is obtained by setting all the ``*" entries in $M$ to zero making it  the identity matrix (see Appendix~\ref{App:Nuclear}).  Next, we will show that other rank minimization methods, such as the alternating projections (AP) method \cite{fazel2004rank, AlternatingBoyd2003}, can be used to construct near-optimal  scalar linear index codes.

%We expect that a careful analysis of the tools developed in the matrix completion literature, especially the recent results in \cite{recht2011simpler, jain2013low}, will allow to incorporate the structure of the family of index coding matrices which  can help to answer   Problem~\ref{OP:CVX2}. 

\subsection {Alternating Projection Method}

Given two convex regions $\mathcal{C}$ and $\mathcal{D}$, a sequence of alternating projections  between these two regions converges to a point in their intersection as illustrated in Fig.~\ref{fig:AP}(a) \cite{Bre67, AlternatingBoyd2003, fazel2004rank}.
Therefore, completing the index coding matrix $M$ by choosing values for the ``*" such that $M$ has a low rank $r$ can be thought of as finding the intersection of two regions $\mathcal{C}$ and $\mathcal{D}$ in $\mathbb{R}^{n\times n}$, in which 
\begin{equation}\mathcal{C}=\{M\in \mathbb{R}^{n\times n};\rank(M)\leq r \},\label{eq:regionC}\end{equation}
is the set of matrices of rank less or equal to a given  rank $r$, and 
\begin{multline}
\mathcal{D}=\{M \in \mathbb{R}^{n\times n}; m_{ij}=0 \text{ if } (i,j)\notin G \ \text{ and } m_{ii}=1,\\ 
 i=1,\dots,n  \}.\label{eq:regionD}
\end{multline}
 Note that  $\mathcal{C}$ is not convex  and therefore convergence of the AP method is not guaranteed. However, the AP method can give a certificate, which is the completed matrix $M$, that a certain rank $r$ is achievable. Therefore, we will use the AP method as a    heuristic as described in algorithm~\ref{alg:AP}.

 %%%%%%%%%%%%%Algorithm%%%%%%%%%%%%%%%%%%%% 
 \begin{algorithm}[b!]
 \SetAlgoRefName{APIndexCoding}
 \KwIn{Graph $G$ (or $G_d$)}
 \KwOut{Completed matrix $M^*$ with low rank $r^*$}
Set $r_k=$ greedy coloring number of $\bar{G}$\; \nllabel{Proj:1}
\While{$\exists M \in \mathcal{C'}$  such that $\rank M\leq r_k$\ \ }{
Randomly pick $ M_0 \in \mathcal{C'}$. Set $i=0$ and $r_k=r_k-1$\;
 \Repeat{$\| M_{i+1}- M_{i}\| \leq \epsilon$\nllabel{Proj:11}} 
{
 $i=i+1$\;
\tcc{Projection on $\mathcal{C'}$ (resp. $\mathcal{C}$)  via  eigenvalue decomposition (resp. SVD)} 
Find the eigenvalue decomposition $M_{i-1}=U\Sigma V^T,$ with $\Sigma =$diag$(\sigma_1, \dots, \sigma_n ), \ \sigma_1\geq \dots \geq \sigma_n$\;
Set  $\sigma_l=0$ if $\sigma_l<0,$ $l=1,\dots,n$\;
Compute ${M_i}=\sum_{j=1}^{r_k} \sigma_j u_j v_j^T$\;\nllabel{Proj:C}
 \tcc{Projection on $\mathcal{D}$ } 
$M_{i+1} ={M_i}$
Set the diagonal  entries of $M_{i+1}$ to $1$'s\; \nllabel{Proj:D1}
Change the $(a,b)^{th}$ position in $M_{i+1}$ to $0$ if edge $(a,b)$ does not exist in $G$\; \nllabel{Proj:D2}
} 
%Set {\em Rnk} equal to the number of singular values of $M_{i+1}$ that are larger than $\epsilon$\;
 }
\Return $M^*={M_i}$ and $r^*=r_k.$
 \caption{Alternating projections method for index coding.} \label{alg:AP}
\end{algorithm}

\noindent{\bf Algorithm~\ref{alg:AP}:}  The projection of  a matrix on  the region $\mathcal{C}$ is obtained by singular value decomposition (SVD) \cite{Eckart36}. We noticed from our simulations that a considerable improvement in  performance and    convergence rate,  (See Figs.~\ref{fig:SVDundirRmin} and \ref{fig:SVDundirTime} in Appendix~\ref{App:fig})  can be obtained by projecting on $\mathcal{C}'\subseteq \mathcal{C}$, the set of  positive semi-definite matrices of rank less or equal than $r$, 
\begin{equation}
\mathcal{C}'=\{M\in \mathbb{R}^{n\times n}; M\succeq 0 \text{ and } \rank(M)\leq r \}.\end{equation}

The projection on $\mathcal{C}'$ is obtained by   eigenvalue decomposition and taking the eigenvectors corresponding to the $r$ largest eigenvalues, as done in Step~\ref{Proj:C}. The Projection on region $\mathcal{D}$ is obtained by setting the diagonal entries of the matrix to $1$ and the $ab$th entry to $0$ if edge $(a,b)$ does not exist in $G$, as done in Step~\ref{Proj:D1} and~\ref{Proj:D2}.

%
%
%We define the region $\mathcal{C}'\subseteq \mathcal{C}$ to be the set of  positive semi-definite matrices of rank less or equal than $r$, 
%$$\mathcal{C}'=\{M\in \mathbb{R}^{n\times n}; M\succeq 0 \text{ and } \rank(M)\leq r \}.$$
%In Algorithm~\ref{alg:AP}, we project on $\mathcal{C}'\subseteq \mathcal{C}$ instead of $\mathcal{C}$  since this leads to considerable speedup   in convergence as we noticed from our simulations {\red ( See Fig.~\ref{fig:SVDundirTime})}.
Theoretically, the time complexity of the algorithm can be reduced by doing a binary search on $r$. However, we found that it is much faster to  start with $r$ equal to the  coloring number returned by the greedy coloring algorithm (Step~\ref{Proj:1}).
The stopping criteria in   Step~\ref{Proj:11} uses the $\ell^2$ norm, $\left \| \cdot \right \|$, which is equal to the largest singular value of the matrix. 
%$\| M_{i+1}- M_{i}\| \leq \epsilon$ means that the number of singular values of $M_{i+1}$ that are larger than $\epsilon$ is less or equal than $r_k$. In the whole paper, we use $\left \| \cdot \right \|$ to denote the $\ell^2$ norm equal to the largest singular value of a matrix. }

 \subsection{Simulation Results}
 
 %%%%%%%%%%%%%Figure%%%%%%%%%%%%%%%%%%%% 
 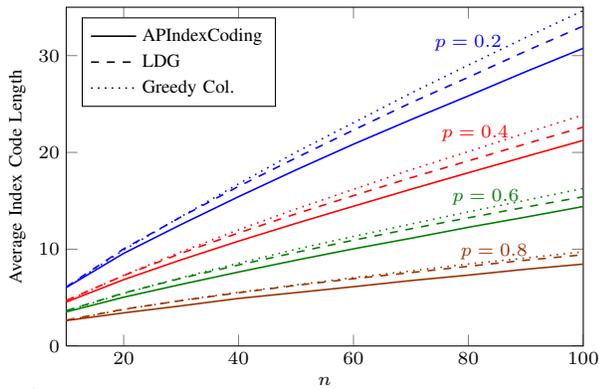
\begin{figure}[h!]  \centering \scriptsize
 \setlength\figureheight{5cm} 
\setlength\figurewidth{7.2cm}
\input{AltPrj_LDG_GreedyColor_iter=1000.tikz}\vspace{-0.3cm}
\caption{\footnotesize Average index code length obtained by APIndexCoding, LDG and  Greedy Coloring on random undirected graphs $G(n,p)$.}\label{fig:APav}

 %\draw[thick] (5.3,3.8) ellipse (0.12cm and 0.3cm);
%  \draw[thick] (5.32,2.8) ellipse (0.12cm and 0.3cm);
%  \draw[thick] (5.32,2.1) ellipse (0.12cm and 0.3cm);
%\draw[thick] (5.4,1.4) ellipse (0.12cm and 0.3cm);

\end{figure}

 We  tested the performance of algorithm~\ref{alg:AP}  on randomly generated graphs. We used the Erdos-Renyi model to generate random undirected graphs $G(n,p)$ on $n$ vertices where edges between two vertices are chosen iid with probability $p$. We compared the performance of algorithm ~\ref{alg:AP} to greedy coloring\footnote{We used the greedy coloring function in the mathgraph Matlab Library.} and Least Difference Greedy (LDG) (see Appendix~\ref{App:LDG} for details on LDG). We also tested the Alternating Minimization method (AltMin) \cite{fazel2004rank, jain2013low, Hardt2014} described in  Appendix~\ref{App:AM}. It does not perform as good as AP (see Fig.~\ref{fig:AMp02})  and suffers from a  slow convergence rate. 
 %(see Fig.~\ref{fig:AMp02Time}). 

  %%%%%%%%%%%%%Figure%%%%%%%%%%%%%%%%%%%% 

 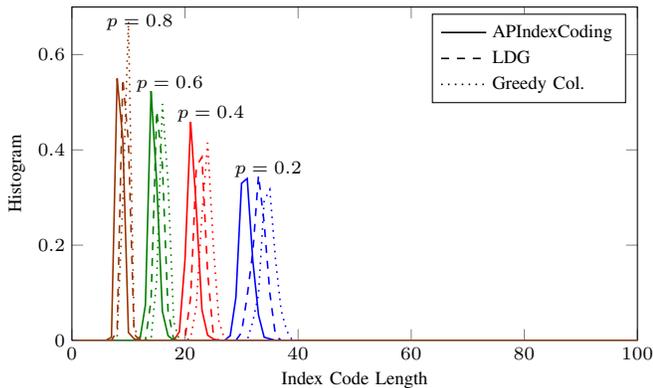
\begin{figure}[t!] \centering\scriptsize
 \setlength\figureheight{5cm} 
\setlength\figurewidth{8cm}
\input{AltPrj_pdf_n=100_iter=300.tikz} \vspace{-8pt}
\caption{ \footnotesize  Histogram of index code length  obtained by APIndexCoding, LDG and  Greedy Coloring on random undirected graphs $G(n,p)$ with $n=100$.}\label{fig:Hist}
\end{figure}  

  Fig.~\ref{fig:APav} shows the average rank obtained by the APIndexCoding algorithm for $n$ between $0$ and $100$ and different values of $p$. In all our simulations,   each data point is obtained by running the algorithms on $1000$ graph realizations and $\epsilon=0.001$ in the stopping criterion.
The  APIndexCoding algorithm always outperforms LDG and Greedy coloring. For instance, an improvement  of $13.6\%$ over greedy coloring is obtained for $n=30$ and $p=0.8$. 
Fig.~\ref{fig:Hist} shows the histogram of the distribution of the rank by APIndexCoding which suggests a concentration around the mean of ranks returned by APIndexCoding\footnote{The concentration of the minimum rank of $G(n,p)$ around its average can be proven using  the vertex exposition martingale method  \cite{AloSpe92}. However, finding an expression of the average remains an open problem \cite{haviv2012linear}.}.  Fig.~\ref{fig:APPer1} shows the savings  achieved by APIndexCoding over  linear network codes which allow all users to decode all the messages (multicast)\footnote{Linear network codes can achieve multicast by transmitting $n-\min_i|H_i|.$}. Similarly, Figs.~\ref{fig:APPer2}, \ref{fig:APPer3} and \ref{fig:APPer4} in Appendix~\ref{App:fig} show the percentage savings of APIndexCoding over  uncoded transmissions, greedy coloring and LDG.

{\em Lower bounds:} We tested the APIndexCoding algorithm on all non-homomorphic directed graphs on at most $5$ vertices and compared its performance to the optimal rates  reported in \cite{Kimurl}. APIndexCoding was always able to find the optimal index coding length except for when it is not an integer ($28$ graphs on $n=5$ vertices). Moreover, we tested APIndexCoding on random 3-colorable graphs (3-partite graphs). For these graphs, we know a priori that the matrix $M$ could be completed to have rank $3$ or less. Fig.~\ref{fig:3col} shows that APIndexCoding beats greedy coloring and LDG and gives an average rank very close to $3$.

 \begin{figure}[t!]  \centering \scriptsize
 \setlength\figureheight{4.8cm} 
\setlength\figurewidth{6.8cm}
\input{AltPrj_vs_Multicast_iter=1000.tikz}\vspace{-0.3cm}
\caption{\footnotesize  Savings in  percentage of APIndexCoding over  multicast network codes.}\label{fig:APPer1}
\vspace{-0.2cm}
\end{figure}
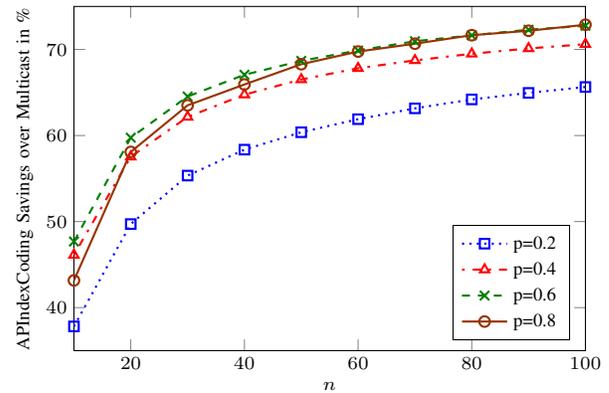

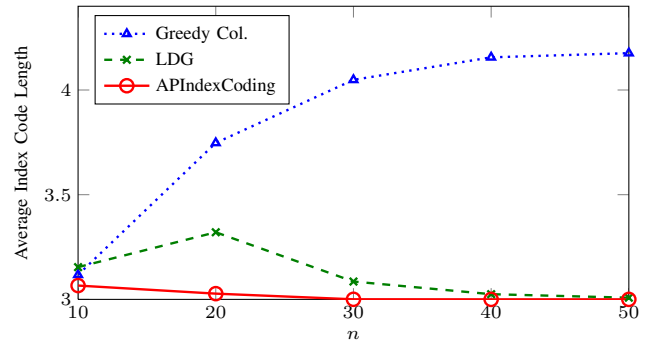
\begin{figure}[t!]
\centering \scriptsize
 \setlength\figureheight{3.9cm} 
\setlength\figurewidth{7.7cm}
      \input{3color_p=0.5_iter=5000_Rmin.tikz}\vspace{-0.2cm}
\caption{\footnotesize Average index code length obtained by using Greedy Coloring, LDG and \ref{alg:AP} for random 3-colorable graphs when $p=0.5$.}
 \label{fig:3col}
\vspace{-0.2cm}
\end{figure}

% {\red MAYBE, WE WANT TO TALK HERE ABOUT INDEX CODING ON RANDOM GRAPHS....}

 %%%%%%%%%%%%%Figure%%%%%%%%%%%%%%%%%%%%

\subsection{Convergence Rate and Running Time}

We ran the simulations on a DELL XPS i7 - 16GB Memory Desktop using Matlab software. Figs.~\ref{fig:undirectedTime} and \ref{fig:undirectedIt}   depict respectively    the average time and average number of iterations taken by the APIndexCoding algorithm to converge on a random undirected graph $G(n,p)$. We notice that the time complexity of the algorithm roughly increases exponentially as $n$ increases (and  $p$ constant)  and as  $p$ increases (and $n$ constant).
% This increase can be explained that as the matrix becomes more dense (the average number of stars is  $n^2p$) the algorithm takes more time to converge. 

To speed up the converge time, we tested a variant of the AP method, called Directional Alternating Projections (DirAP) \cite{Control1987} which is  described in Appendix~\ref{App:DAP}. DirAP can lead to considerable savings in time as seen in Fig.~\ref{fig:TimeDir} ($60\%$  for $n=10$ and $85\%$ for $n=140$, both for  $p=0.2$). We should mention  that Greedy coloring and LDG have complexity quadratic in $n$ and are therefore   much faster than Directional APIndexCoding as seen in Fig.~\ref{fig:TimeDir}. 
However,  the savings in transmissions induced by DirAP or APIndexCoding may justify their  computation overhead in scenarios where the computation can be done offline or  can be amortized over a long time such as finding codes for interference networks with   static or slowly changing topologies.

 %%%%%%%%%%%%%%Figure9%%%%%%%%%%%%%%%%%%%% 
 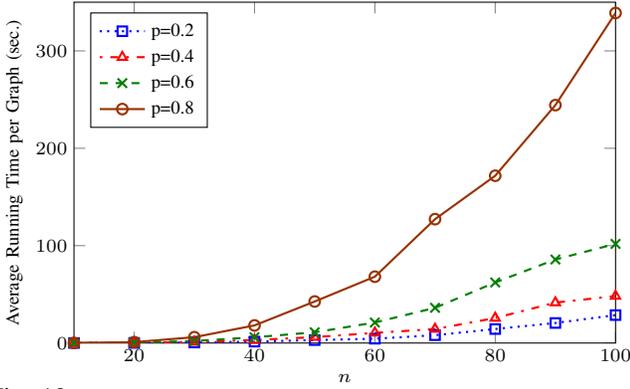
\begin{figure}[h!] \centering \scriptsize
 \setlength\figureheight{4.7cm} 
\setlength\figurewidth{7.2cm}
\input{AltPrj_Time_consumption_linear.tikz}
\vspace{-8pt}
\caption{\footnotesize Average running time of one Graph by using APIndexCoding on random undirected graphs.}
\label{fig:undirectedTime}
 \end{figure}

%%%%%%%%%%%%%Figure%%%%%%%%%%%%%%%%%%%% 
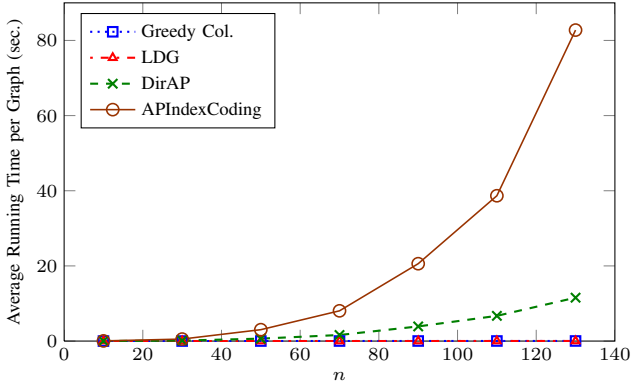
\begin{figure}[h!]  \centering \scriptsize
      \input{dirAP_p=0.2_iter=1000_Time.tikz}
\vspace{-5pt}
\caption{\footnotesize Running time of APIndexCoding and Directional APIndexCoding (DirAP) on random undirected graphs $G(n,p)$ with $p=0.2$.}\label{fig:TimeDir}
 \vspace{-.1cm}
 \end{figure}

\subsection{Decoding Error Analysis}

The \ref{alg:AP} algorithm  returns a completed  matrix $M^*$ with low rank $r^*$. However,  $M^*$  is not in $\mathcal{C}$ in general, but is very ``close" to a matrix in  $\mathcal{C}$ (in $\ell_2$ norm distance) as dictated by   to the stopping criteria of the algorithm.  This will  cause a small decoding error at the users side.

\begin{example}\label{ex:decoding}
For the index coding instance of Fig.~\ref{fig:ex1}, our implementation of  algorithm~\ref{alg:AP} with $\epsilon=0.001$ returns  following matrix $M^*$ with rank $2$,

\begin{equation}
M^*=\begin{bmatrix*}   
    1.0000  &  1.4492  &  1.8671 &  1\cdot 10^{-5} \\
    0.6900  &  1.0000 &   1.2883  & -1 \cdot 10^{-5} \\
 9 \cdot 10^{-6}  &  0.7762   & 1.0000  & -0.7519  \\
    0.7122   &  1 \cdot 10^{-5}    & -1 \cdot 10^{-5}  &  1.0000  \\
\end{bmatrix*}.\label{eq:matrixsimu}
\end{equation}
   It can be seen that $M^*$ is not in $\mathcal{C}$ since that the positions that are supposed to be zero are not exactly $0$ but relatively  small numbers. 
\end{example}

The next result shows that  if the quantization interval  of the messages $X_i$'s is not very small,  the decoding error can be avoided. Assume $X_i\in[-X_{max},X_{max}], \  i=1, \dots,n$, and let $\hat{X}_i$ be the decoded message $X_i$. Lemma~\ref{lemma:error} upper bounds the decoding error as a function of $\epsilon$, where $\epsilon$ is the distance of the matrix $M^*$ to $\mathcal{C}$ and is used as a   stopping criteria in \ref{alg:AP}.

\begin{lemma}\label{lemma:error}
Let $\mathbf{X}=[X_1,X_2, \dots, X_n]^T$ be the message vector at the transmitter. Assume that the index code  given by matrix $M^*$ is used  and let $ \hat{\mathbf{X}}=[\hat{X}_1,\hat{X}_2, \dots, \hat{X}_n]^T$ be the messages decoded by the users. Then, 

\begin{equation}
\| \mathbf{X} - \hat{\mathbf{X}}\| \leq \epsilon X_{\max}\sqrt{n}.
 \end{equation}
\end{lemma}
\begin{proof}
See Appendix~\ref{App:ERROR}.
\end{proof}

  To illustrate the result in Lemma~\ref{lemma:error}, we first elaborate on  the encoding and decoding functions of the index code once $M^*$ is obtained from the algorithms.   
  
 Let $r^*$ be the rank of $M^*$ and  Let   $A$ be  a $r^*\times n$ submatrix of $M^*$  of rank $r^*$.  WLOG, we  can assume that $A$ is formed of the  first $r^*$ rows of $M^*$. 
Let $\underbar{m}_i^{*}$ denotes the $i$th row of $M^*$, with $i=1, \dots, n$.
 
 %and denote them by $B=\begin{bmatrix*} \psi_{1}^T & \psi_{2}^T  & \cdots & \psi_{r^*}^T  \end{bmatrix*}^T $.

  The transmitter  broadcasts 
 \begin{equation}
 \mathbf{Y}=\begin{bmatrix*}  Y_1 , Y_2 , \dots , Y_{r^*} \end{bmatrix*}^T=A \mathbf{X}.
 \end{equation}
  % where  $\mathcal{X}= \begin{bmatrix*}  X_1 , X_2 , \dots , X_{n} \end{bmatrix*}^T$ is the vector of all the messages.

When decoding, user $i$ can obtain  $\hat{X_i}$ by the following decoding equation: 
\begin{equation}
\hat{X_i}=\left\{\begin{matrix}
Y_i - \underbar{m}_i^{*} \phi_{i}^T  & 1\leq i \leq r^*, \\
\underbar{m}_i^{*} A^\dagger \mathcal{Y}- \underbar{m}_i^{*} \phi_{i}^T &  r^*<  i \leq n,
\end{matrix}\right. \label{eq:decode} 
\end{equation}
where $A^\dagger=A^T(AA^T)^{-1}$  is the Moore-Penrose pseudoinverse of $A$ and  $\phi_{i}$ is a dimension $n$ vector that contains all  the side information that user $i$ has, with $0$'s on all the other positions. For instance, in the example of Fig.~\ref{fig:ex1}, $\phi_1 =\begin{bmatrix*}  0 , X_2 , X_3 , 0 \end{bmatrix*}.$

 {\em Example 2 (continued):}  Suppose the  transmitter wants to send $\mathbf{X}=\begin{bmatrix*} 10,10 ,-10, 10 \end{bmatrix*}^T$ to the users. Let $$A=\begin{bmatrix*}   
    1.0000  &  1.4492  &  1.8671 &  1\cdot 10^{-5} \\
 9 \cdot 10^{-6}  &  0.7762   & 1.0000  & -0.7519 
\end{bmatrix*}$$ be a submatrix of $M^*$ in \eqref{eq:matrixsimu} of rank $2$. Then, the transmitter should broadcast
 $\mathbf{Y} = A \mathbf{X}= \begin{bmatrix*}  5.8211 , -9.7575  \end{bmatrix*}.$
%   then it should encode $\mathcal{X}$  into $2$ broadcast messages, $\{ 10.0004 +  12.5308  - 15.0566  -3.4281\cdot 10^{-4} ,      7.9822  +  10.0017 -  12.0177   -1.2699 \cdot 10^{-4}  \}$, that is, $\mathcal{Y}_B = \begin{bmatrix*}  7.4744 , 5.9661  \end{bmatrix*}$.
The decoding vector given by \eqref{eq:decode}  is 
     $$\hat{\mathbf{X}}= \begin{bmatrix*} 9.999, 9.9998 , -9.9997, 10.0002 \end{bmatrix*}^T .$$
The aggregate decoding error here is  $\| \mathbf{X}-\hat{\mathbf{X}} \| =3.6894 \cdot 10^{-4}$. This should be compared to the   bound from Lemma~\ref{lemma:error} which gives $\| \mathbf{X}-\hat{\mathbf{X}} \| =2 \cdot 10^{-2}.$

In general, it would be interesting to  bound the decoding error per user. However, we found it more tractable to bound the aggregate decoding error. The bound on the decoding error in Lemma~\ref{lemma:error} is  loose, but can give guidelines on how the stopping criteria affects the decoding error and can help design the quantization of the source if zero-decoding error is required. Fig.~\ref{fig:EError} shows the gap between  the theoretical bound of Lemma~\ref{lemma:error} and the average error obtained in simulations. 

  %%%%%%%%%%%%%Figure%%%%%%%%%%%%%%%%%%%% 
\begin{figure}[h] \centering \scriptsize
 \setlength\figureheight{4.4cm} 
\setlength\figurewidth{7.2cm}
      \input{ErrorFactor_p=0.2_iter=1000.tikz} \vspace{-.2cm}
\caption{\footnotesize Average decoding error  $\| \mathbf{X}-\hat{\mathbf{X}} \|$ in  APIndexcoding on random undirected graphs when $p=0.2$, $\epsilon=0.001$ and $X_i\in[-10,10]$ ($X_{\max}=10$).}\label{fig:EError}
 \vspace{-.2cm}
 \end{figure}
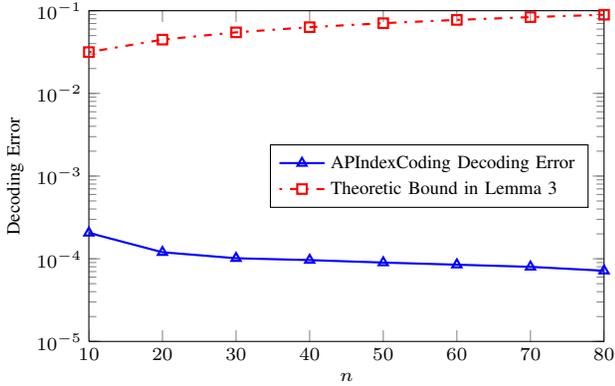

%%%%%%%%%%%%%%Figure%%%%%%%%%%%%%%%%
\section{Index Coding on Random Directed Graphs}\label{sec:dir}
 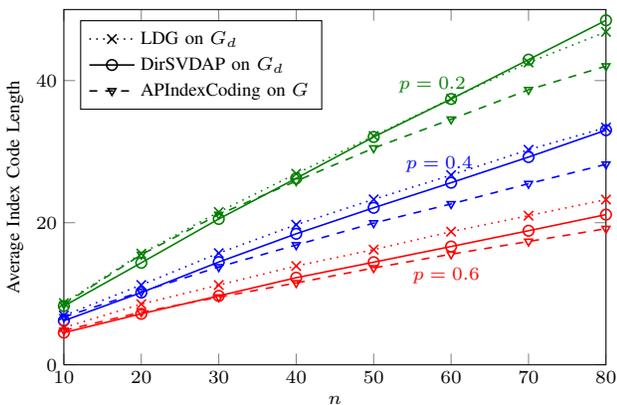
\begin{figure}[h] \centering\scriptsize
 \setlength\figureheight{5.8cm} 
\setlength\figurewidth{9.3cm}
\input{Directed_DirSVDAP_iter=100_Rmin.tikz}
\vspace{-8pt}
\caption{\footnotesize Average index code length of LDG, Directional APIndexCoding via SVD (DirSVDAP) on $G_d$ and  APIndexCoding on the undirected subgraph $G$,  for random directed graphs $G_d(n,p)$.}\label{fig:dirSVDAP}
\vspace{-8pt}
\end{figure}
%\draw[thick] (5.9,3.9) ellipse (0.13cm and 0.3cm);
%\draw[thick] (5.9,2.9) ellipse (0.13cm and 0.3cm);
%\draw[thick] (5.9,2.2) ellipse (0.13cm and 0.3cm);

In this section, we consider the more general case in which  the side information graph $G_d$ is a directed random graph. Each directed edge $(i,j)$ exists with probability $p$ and the graph edges are chosen independently.  In this case, we can apply all the rank minimization methods described in the previous section on the graph $G$, the maximal undirected subgraph of $G_d$. In addition, we can apply the AP and Directional AP methods via SVD directly on the graph $G_d$ (SVD is needed here because the matrix $M$ is not symmetric).
Fig.~\ref{fig:dirSVDAP} depicts the { top three among these methods having the best performance}. For relatively small values of $n$, DirSVDAP on the directed graph $G_d$ has the best performance. However, for large $n$ APIndexCoding on $G$ performs better. It is worth mentioning that this directed random graph model was used  in  Fig.~\ref{fig:simu1} and the results on alternating projections there were also obtained by applying APIndexCoding algorithm on $G$.

To address the practical setting in which users have a fixed cache size, we evaluated the performance of these methods on   random directed regular graphs. These results are presented in Fig.~\ref{fig:cache} and show that APIndex coding gives the best results in terms of minimizing the index code length.

%%%%%%%%%%%%%%Figure%%%%%%%%%%%%%%%%

\section{Network coding via Rank Minimization}\label{sec:RnkMin}

Network coding can be thought as a generalization of routing schemes in networks. It allows intermediate nodes  to forward coded packets that are functions of their incoming packets \cite{EminaMonograph, ho2008network,Y10}. There are now efficient algorithms to  construct  capacity-achieving network codes for multicast networks and some related variants \cite{JagSanCEEJT, KM03, HoMKKESL06, lun2006minimum}. However,  making similar progress for general networks with non-multicast demands is believed to be a very hard problem \cite{DFZ05,DFZ07,Lehman2004complexity, medard2003coding}, even for two-unicast networks \cite{Sudeep}. With this backdrop, the rank minimization heuristics presented here provide a computational tool that  contribute to making progress towards constructing  linear network codes for non-multicast networks.

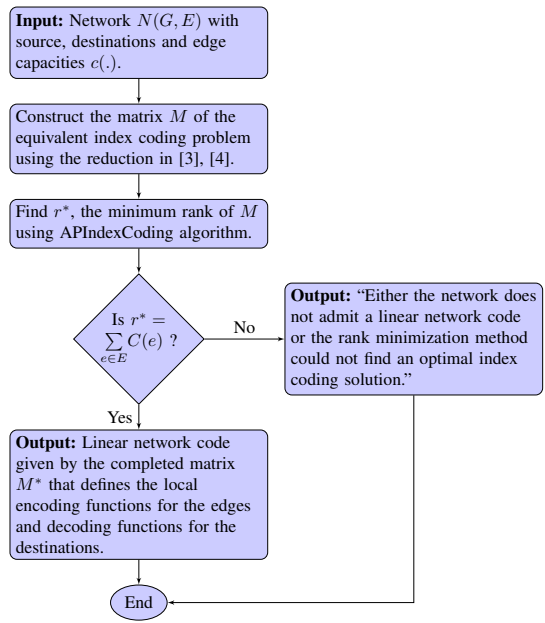
\begin{figure}[h!]
\begin{center}

% Define block styles
\tikzstyle{decision} = [diamond, draw, fill=blue!20, 
    text width=4.5em, text badly centered, node distance=3cm, inner sep=0pt]
\tikzstyle{block} = [rectangle, draw, fill=blue!20, 
    text width=14em, rounded corners, minimum height=2em]
\tikzstyle{line} = [draw, -latex']
\tikzstyle{cloud} = [draw, ellipse,fill=blue!20, node distance=3cm, minimum height=2em]
\resizebox{0.4\textwidth}{!}{
    
\begin{tikzpicture}[node distance = 2cm, auto]
    % Place nodes
    \node [block] (input) {{\bf Input:} Network $N(G,E)$ with source, destinations and edge capacities $c(.)$.};
%    \node [cloud, right of=input] (system) {system};
    \node [block, below = 0.5 cm of input] (matrix) {Construct the matrix $M$ of the equivalent index coding problem using the reduction in  \cite{RSG10, EEL13}.};
    \node [block, below =0.5 cm of matrix] (APIndexCoding) {Find $r^*$, the minimum rank of $M$ using APIndexCoding algorithm.};
    \node [decision, below =0.5cm of APIndexCoding] (decide) {Is $r^*=\underset{e\in E}{\sum} C(e)$ ?};
    \node [block, below = 0.5cm of decide, node distance=4.5cm] (Yes) {{\bf Output:} Linear network code given by the  completed matrix $M^*$ that defines the local encoding functions for the edges and decoding functions for the destinations.};
    \node [block, right of=decide, node distance=5.5cm] (No) {{\bf Output:} ``Either the network does not admit a linear network code or the rank minimization method could not  find an  optimal index coding solution."};
    \node [cloud, below =0.5cm of Yes, node distance=2cm] (end) {End};
    % Draw edges
    \path [line] (input) -- (matrix);
    \path [line] (matrix) -- (APIndexCoding);
    \path [line] (APIndexCoding) -- (decide);
    \path [line] (decide) -- node [midway,left] {Yes} (Yes);
    \path [line] (decide) -- node [midway,above] {No} (No);
    \path [line] (Yes) -- (end);
    \path [line] (No) |- (end);
\end{tikzpicture}
}
 \vspace{-0.1cm}
\caption{ \footnotesize Flowchart summarizing the different steps in our code in \cite{APlink} for constructing linear network codes for general networks using rank minimization via APIndexCoding algorithm. }\label{fig:flowchart}
\end{center}
\vspace{-0.5cm}
\end{figure}
 
 The main idea here is to use the efficient reduction in \cite{RSG10, EEL13} to transform a given network coding problem $\mathcal{NC}$ to an index coding  problem $\mathcal{IC}$ and then to apply the rank minimization methods presented here to $\mathcal{IC}$. 
Suppose that $\mathcal{NC}$ is defined over a network $N(V,E)$ with vertex set $V$,  edge set $E$, and each edge $e\in E$ has capacity $c(e)$. The reduction guarantees the following property: {\em $\mathcal{NC}$ has a network code over a certain alphabet that allows all the destinations to decode their messages with zero probability of error  if and only if $\mathcal{IC}$ has an index code of length $r^*=\sum_{e\in E} c(e)$ over the same alphabet.} This property gives the algorithm  illustrated in Fig.~\ref{fig:flowchart}. The proposed rank minimization methods are not guaranteed to find the minimum rank (i.e., minimum scalar  linear index code length), but can give a certificate (the completed matrix)  for the low rank it find. For this reason, the algorithm in Fig.~\ref{fig:flowchart} either outputs a linear network code solution or a ``do not know"  message.   This algorithm was implemented in Matlab and can be found and tested on the link in \cite{APlink}.

\section{Conclusion}\label{sec:conclusion}
%SUMMARIZE OBSERVATION ABOUT GRAPHS
%
%TALK ABOUT LOWER BOUND
%
%
%INDEX CODING ON RANDOM GRAPHS
%
%
%OPEN PROBLEMS

We have investigated  the performance of different rank minimization methods for constructing  linear index codes over the reals. Our simulation results indicate that the Alternating Projections method and its directional variant, always 
outperform (smaller code length) graph coloring algorithms, and they converge much faster than the Alternating Minimization method. Due to the special structure of the underlying matrices representing the index coding problem (all ones diagonal), the well-studied  nuclear norm minimization method performs badly here. Our results lead to the following open questions that we plan to address in our future work:

\begin{enumerate}
\item Can the  proposed methods here be adapted to construct linear index codes over finite fields?
\item Under what conditions on the index coding matrices, can these methods be given theoretical guaranties to construct optimal linear index codes?
 \end{enumerate}

%{\red We investigate the performance of different methods for constructing linear index codes. We recommend to use directional APIndexCoding Algorithm for both undirected and directed graphs. For undirected graphs, APIndexCoding Algorithm could always outperforms then other Algorithms in the literature. For undirected graphs, APIndexCoding Algorithm could outperforms when $n$ and $p$ are not too small.
%
%There are many open problems. Simulations shows that APIndexCoding Algorithm could find $L_{min}$ easily up to $n=5.$}

\section{Acknowldgement}
The second  author would like to thank Prof.  Stephen Boyd for suggesting the use of the alternating projections method,  Borja Peleato-Inarrea and Carlos Fernandez for discussions on the alternating minimization method and Alex Dimakis for insightful discussions on index coding and graph coloring.

\begin{appendices}

%%%%%%%%%%%%%%%%%%%%%%%%%%%%%%%%%%%%%%%%%%%%%%%%%%%%%%%%%%%
\section{Nuclear norm minimization}\label{App:Nuclear}
Using the nuclear norm minimization method to minimize the rank of the index coding matrix $M$ will always give the maximum rank $n$. This corresponds to the trivial index code obtained by replacing all the ``*" in $M$ by zero. This follows directly from the results in  \cite{recht2010guaranteed}  which we reproduce here for completion.
 Let ${\left \| \cdot \right \|}_*$ denotes the nuclear norm.
 %and ${\left \| \cdot \right \|}$ denotes the operator norm equal to the largest singular value of a matrix. 
%It was shown in  \cite{recht2010guaranteed} that the nuclear norm ${\left \| \cdot \right \|}_*$ is equal to the dual norm ${\left \| \cdot \right \|}_d$ of the operator norm ${\left \| \cdot \right \|}$ in $\mathbb{R}$, where the operator norm ${\left \| \cdot \right \|}$ of a matrix is equal to its largest singular value. A dual norm  ${\left \| M \right \|}_d$ is defined as the maximum  trace of $M'X$, given $X$ is a matrix with the same size of $M$ and singular values at most $1.$ So we have the following lemma from \cite{recht2010guaranteed}.
\begin{lemma}
The nuclear norm can be written as,
$${\left \|  M \right \|}_* = \max \{ \Tr(M^TX); X\in \mathbb{R}^{n\times n},\left \|  X \right \| \leq 1 \}, $$
 where $\Tr(.)$ is the trace of a matrix. 
 \end{lemma} 
In the previous lemma, if we pick $X$ to be the identity matrix then $ {\left \|  M \right \|}_*  \geq   \Tr(M) = n. $ Therefore, applying the  nuclear norm minimization to the index coding problem will always return the diagonal matrix as the optimal solution.

%%%%%%%%%%%%%%%%%%%%%%%%%%%%%%%%%%%%%%%%%%%%%%%%%%%%%%%%%%%
\section{LDG Algorithm:}   \label{App:LDG}

Birk and Kol proposed a greedy algorithm named \emph{Least Difference Greedy (LDG)} in  \cite{BirkKol,ISCOD2006} for finding scalar linear index codes. LDG can be regarded as a heuristic for finding  clique cover for graphs. The idea is to minimize the rank of the index coding matrix $M$ by greedily searching for rows that could be made equal and ``merging" them.
Two rows are mergeable if there does not exist any column in which one of these rows has a  $``0"$ and the other a $``1"$. Therefore,  the two rows can be made the  same by giving appropriate values to  $``*"$.  
For instance, in the example of Fig.~\ref{fig:ex1}, we start from matrix
\begin{equation*}
M_0=\begin{tabular}{m{8pt} m{8pt} m{5pt} m{3pt}m{3pt}}
 & $X_1$ & $X_2$ & $X_3$ & $X_4$ \\         
 $row_1 $ & \multicolumn{4}{c}{\multirow{4}{*}{$\begin{pmatrix}
 1 & * & * & 0\\ 
 * & 1 & * & 0\\ 
 0 & * & 1 & *\\ 
 * & 0 & 0 & 1
\end{pmatrix}$}} \\
$row_2 $& \\
$row_3 $&  \\
$row_4$ & \\
\end{tabular}
\end{equation*}
Row $1$ and row $2$ are mergeable. After  merging  them we get 
\begin{equation*}
M_1=\begin{tabular}{m{8pt} m{8pt} m{5pt} m{3pt}m{3pt}}
 & $X_1$ & $X_2$ & $X_3$ & $X_4$ \\         
 $row_1$ & \multicolumn{4}{c}{\multirow{3}{*}{$\begin{pmatrix}
 1 & 1 & * & 0\\
 0 & * & 1 & *\\ 
 * & 0 & 0 & 1
\end{pmatrix}$}} \\
$row_3 $&  \\
$row_4$ & \\
\end{tabular}
\end{equation*}
There are no more mergeable rows in $M_1$. The remaining ``*'''s can be set arbitrarily, for example they could be set all to $0$.  And, the LDG algorithm will output the $3$ transmitted messages  $X_1+X_2$, $X_3$ and $X_4$. For completion, we give   next the details of the  LDG algorithm as proposed in \cite{BirkKol,ISCOD2006}.

\begin{algorithm}[h!]
 \SetAlgoRefName{LDG}
 \KwIn{Index coding $n\times n$ matrix $M$.}
 \KwOut{Linear index code over $GF(2)$.}
  Set $i=1$\;
 \While{$i<n$}{
 Row set $\mathcal{S} :=\{row_{i+1},\cdots, row_{n} \}$ \;
 \While{$\exists $ at least one row in $\mathcal{S}$ mergeable with $row_i$}{
Randomly pick a mergeable row $row_j$ from  $\mathcal{S}$\;
Merge $row_j$ into $row_i$ column by column by using the following rules: $``*"+``*"=``*",  \  1+``*"=1,  \  0+``*"=0$\;
Delete $row_j$ from matrix $M$\;
 }
 $i=i+1$ \;
 }
 \ForAll{$row_i$ in $M$ }{
Create a coded message by XORing all messages that corresponding to the positions  of $1$ in $row_i$\;
   }
   % $row_a,$ $\{row set\} $
%\Function{test}{$row_a$, $\{row set\} $}{
%
\caption{The Least Difference Greedy Clique-Cover method~\cite{BirkKol,ISCOD2006}. } \label{alg:LDG}
\end{algorithm}

%%%%%%%%%%%%%%%%%%%%%%%%%%%%%%%%%%%%%%%%%%%%%%%%%%%%%%%%%%%
\section{Alternating Minimization Algorithm:} \label{App:AM}

%\noindent{\bf Algorithm Alternating Minimization (\ref{alg:AltMin}):} 

The Alternating Minimization (AltMin) is now a well studied method for rank minimization~\cite{fazel2004rank, jain2013low, Hardt2014}. We briefly describe it here for completion.

% Our simulations show that AltMin does not perform as good as AP, and converges much slower.  
If the  matrix $M \in \mathbb{R}^{n \times n}$  has rank $r$ then it can be factored as  $M=EF^T,$ where $E$ and $F \in \mathbb{R}^{n \times r}$ and $F^T$ is the transpose of $F$. Thus, the problem becomes the following 
\begin{equation}\underset {M\in \mathcal{C}, \  E,  F \in \mathbb{R} ^{n \times {r}} }  {\text {argmin}}||M- EF^T||_F,\label{eq:AM}\end{equation}
where $\| \cdot \|_F$ denotes the Frobenius norm.  The optimization problem in \eqref{eq:AM} is not convex. However, it will become convex if either $E$ or $F$ is fixed. Algorithm~\ref{alg:AM} \cite{fazel2004rank} shows the  iterations between  fixing $E$ and $F$  and solves the resulting convex problem at each time (Steps 5 and 6, respectively). Each of these steps is a least squares problem that has an analytical solution~\cite[p.  4-5]{BoydCVX}.

\begin{algorithm}[h!]
\SetAlgoRefName{AltMin}
 \KwIn{Matrix $M$}
 \KwOut{ Competed matrix $M^*$ with low rank $r^*$}
 Set $r_0=n$ \;
 \While{$\exists M \in \mathcal{C}$  such that $\rank M \leq r_k$:}{
Randomly pick $E_0 \in \mathbb{R} ^{n \times r}$. Set $i=1$ and $r_k=r_k-1$\; \nllabel{Proj:MIN1} 
\Repeat{$|e_i-e_{i-1}|\leq \epsilon,$ or $e_i \leq \epsilon$ }{
   $F_i=\underset {M\in \mathcal{C}, \ F \in \mathbb{R} ^{n \times {r_k}} }  {\text {argmin}}||M- E_{i-1}F^T||_F ,$ \nllabel{Proj:MIN2} \\
 $(M_i, E_i) =\underset {M\in \mathcal{C}, \ E \in \mathbb{R} ^{n \times{ r_k}} }  {\text {argmin}}||M- EF_i^T||_F ,$ \\
 $e_i=||M_i- E_i F_i^T||_F$ \;
  $i=i+1;$ }
 }
\Return $M^*=E_iF_i^T$ and $r^*=r_k$.
\caption{Alternating Minimization~\cite{fazel2004rank}.} \label{alg:AM}
\end{algorithm}

%%%%%%%%%%%%%%%%%%%%%%%%%%%%%%%%%%%%%%%%%%%%%%%%%%%%%%%%%%%
\section{Directional Alternating Projections:} \label{App:DAP}
The Directional Alternating Projections (DirAP)  method \cite{Control1987} can  converge faster than AP and can give a low rank close to that of AP.  Fig.~\ref{fig:dirAP}  depicts geometrically  the first steps of the  DirAP method starting with a random point $e_0$ followed by the  first four projection  points  $d_1,c_1,d_2, c_2$. In the AP method, the fourth projection point would be $c_2 \in\mathcal{C}$,  whereas in DirAP the fourth projection is onto the tangent space on $\mathcal{C}$ at point $c_1$ which gives point $e_1$ obtained by the following equation:
 $$e_1=d_1+ \lambda(d_2 - d_1), \text{ with }   \lambda =\frac{\| d_1-c_1 \|_F^2 }{\Tr(d_1-d_2)^T(d_1-c_1)}.$$

It can be shown \cite{Control1987}  that if  $\mathcal{C}$ and $\mathcal{D}$ are two convex regions with intersection, then the series of these projections starting from $e_0, e_1,e_2, \dots$  will converge to a point in $\mathcal{C} \cap \mathcal{D}$.
% \begin{theorem}
% Let $\mathcal{C}$ and $\mathcal{D}$ be two convex regions with intersection. For any point $e_0 \in \mathbb{R}^{n \times n},$ 
%its directional alternating projection must converge to a point in the intersection of the region $\mathcal{C}$ and $\mathcal{D}$.
%The directional projection of $e_0$ is defined as: $$e_1=d_1+ \lambda(d_2 - d_1), \text{ with }   \lambda =\frac{\| d_1-c_1 \|_F^2 }{\Tr(d_1-d_2)^T(d_1-c_1)},$$
%where $d_1$ is the projection of $e_0$ on $\mathcal{D}$, $c_1$ is the projection of $d_1$ on $\mathcal{C}$, and $d_2$ is the projection of $c_1$ on $\mathcal{D}$.\\
%\end{theorem}

%%%%%%%%%%%%%%Figure17%%%%%%%%%%%%%%%%
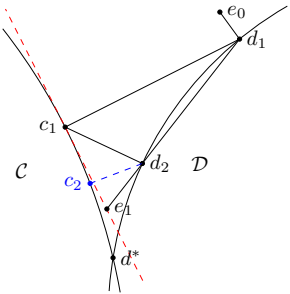
\begin{figure}[h!]
\centering 
\resizebox{0.22\textwidth}{!}{
\begin{tikzpicture}
      \draw(2, 0) arc (10:38:10cm) ;
      \draw(4.8, 4.8) arc (120:176:6cm);
      \filldraw[black] (3.67,4.7) circle (1pt) node[anchor=west] {$e_0$}; 
      \draw(3.67,4.7)--(4,4.247); 
      \filldraw[black] (4,4.247) circle (1pt) node[anchor=west] {$d_{1}$};
      \draw(4,4.247) --(1.078,2.771);
      \filldraw[black] (1.078,2.771) circle (1pt) node[anchor=east] {$c_1$} node [below left  = 0.7 cm] {$\mathcal{C}$};
      \draw [red,dashed] (2.378, 0.1974)--(0.078, 4.7507) ;
      \draw(1.078,2.771)--(2.372, 2.161);
      \filldraw[black] (2.372, 2.161) circle (1pt) node[anchor=west] {$d_2$}node [right  = 0.7 cm] {$\mathcal{D}$};
      \draw[blue,dashed](2.372, 2.161)--(1.496, 1.827);
      \filldraw[blue] (1.496, 1.827) circle (1pt) node[anchor=east] {$c_2$};
      \draw (4,4.247)--(1.777, 1.3986);
      \filldraw[black] (1.777, 1.3986) circle (1pt) node[anchor=west] {$e_1$};
      \filldraw[black] (1.88,0.58)circle (1pt) node[anchor=west] {$d^*$};
\end{tikzpicture}
}
\caption{\footnotesize Directional Alternating Projection (DirAP) method. The projection points starting from a random point $e_0$ are $d_1, c_1, d_2, e_1$. The difference with AP method is that in AP the fourth projection point is $c_2\in \mathcal{C}$ instead of $e_1$.}
 \label{fig:dirAP}
\end{figure}

%%%%%%%%%%%%%%%%%%%%%%%%%%%%%%%%%%%%%%%%%%%%%%%%%%%%%%%%%%%
\section{Proof of Lemma~\ref{lemma:error}} \label{App:ERROR}
The  decoding function in \eqref{eq:decode}, can be rewritten as 
 \begin{equation}
 \hat{\mathbf{X}}=\mathbf{Y} -M^*\circ \Phi  \mathbf{X},
 \end{equation}
 where $\circ$ denotes the entry-wise matrix product (Hadamard product), and $\Phi$ is a $n\times n$ matrix, with $1$'s in the $(i,j)$th positions if edge $(i,j)$ in $G_d,$ and $0$'s in all the other positions. For instance, in the example of Fig.~\ref{fig:ex1}, $$\Phi=\begin{bmatrix*}  0 &1 & 1 & 0\\ 1 & 0 & 1 & 0\\ 0 & 1 & 0 & 1 \\ 1& 0 & 0&0
 \end{bmatrix*}.$$

 We can prove Lemma~\ref{lemma:error} as following:
\begin{align}
\| \mathbf{X}-\hat{\mathbf{X}} \| &= \| \mathbf{X} - M^* A^\dagger A\mathbf{X} -  M^*\circ \Phi  \mathbf{X} \|  \\
&= \| \mathbf{X} - M^*\mathbf{X} - M^*\circ \Phi  \mathbf{X}  \|   \\
&= \| ( I + M^* \circ \Phi  - M^* ) \mathbf{X} \| \\
&= \| ( M_{\mathcal{D}} - M^*  ) \mathbf{X} \| \label{emu:13} \\
&\leq \| M_{\mathcal{D}} - M^*\| \| \mathbf{X} \| \label{emu:14} \\
%&=\| U\begin{bmatrix*}  \Sigma_{r^*}  & 0 \\ 0 & \Sigma_{n-r^*}  \end{bmatrix*}V^T  - U\begin{bmatrix*}  \Sigma_{r^*}  & 0 \\ 0 & 0 \end{bmatrix*}  V^T \|  \| \mathcal{X} \|\\
&=\|U \begin{bmatrix*} 0  & 0 \\ 0 &  \Sigma_{n-r^*} \end{bmatrix*} V^T \| X_{\max} \sqrt{n}\label{emu:15} \\
&\leq  \sigma_{r^*+1} X_{\max} \sqrt{n} \label{emu:16} \\
%&\leq X_{\max}   \sqrt{ n} \| M_{\mathcal{C}} - M^*    \|_F    \\
%&\leq X_{\max}  \sqrt{ n( \sigma_{r^*+1}^2 + \cdots +\sigma_{n}^2) }  \label{equ:10} \\ 
& \leq \epsilon X_{\max}\sqrt{n}.
\end{align}
The matrix $M_{\mathcal{D}}=M_{i-1}$ in \eqref{emu:13} is the matrix   in $\mathcal{D}$  whose projection on $\mathcal{C}$ in the last iteration of APIndexCoding gives the matrix $M^*=M_i$ returned by the algorithm.  Eq~\eqref{emu:15} follows from the fact that if $M_{\mathcal{D}}=U\Sigma V^T$  is the SVD of $M_{\mathcal{D}}$, with $\Sigma = \text{diag}(\sigma_1, \sigma_2,  \dots, \sigma_n )$, then $$M^*=U \begin{bmatrix*} \Sigma_{r^*}   & 0 \\ 0 &  0\end{bmatrix*} V^T.$$
% $$\Sigma = \text{diag}(\sigma_1, \sigma_2,  \dots, \sigma_n )= \begin{bmatrix*}  \Sigma_{r^*}  & 0 \\ 0 & \Sigma_{n-r^*}  \end{bmatrix*}, \ \sigma_1\geq \dots \geq \sigma_n.$$
Eq~\eqref{emu:14} follows from the definition of $\ell_2$ norm $$\|M\| =\sup\limits_{\substack{X\in \mathbb{R}^n\\X\neq {\bf 0}}}\frac{\|MX \|}{\|X \|}.$$
% that is $$\| U \begin{bmatrix*} 0 & \\  & \Sigma_{n-r} \end{bmatrix*}   U^T \|_F  =\sqrt{\sigma_{r+1}^2  + \cdots +\sigma_{n}^2 }.$$

%%%%%%%%%%%%%%%%%%%%%%%%%%%%%%%%%%%%%%%%%%%%%%%%%%%%%%%%%%%

%%%%%%%%%%%%%%Figure19%%%%%%%%%%%%%%%%

\section{Figures} \label{App:fig}
  \setlength\figureheight{3.83cm}

 %%%%%%%%%%%%%Figure%%%%%%%%%%%%%%%%%%%% 
 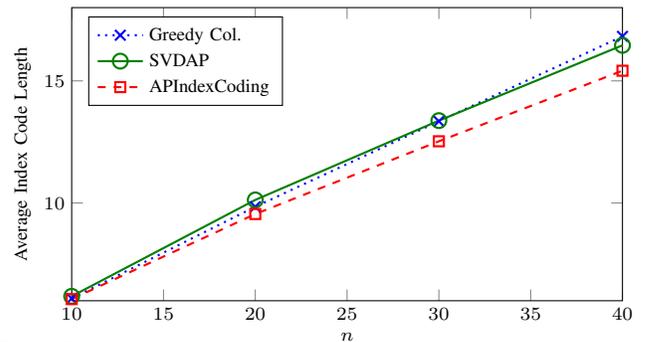
\begin{figure}[h!]
\centering \scriptsize
  \setlength\figureheight{3.9cm} 
\setlength\figurewidth{7.7cm}
      \input{undirected_SVD_p=0.2_iter=500_Rmin.tikz}
\vspace{-0.3cm}
\caption{\footnotesize  Average index code length obtained  using  AP via SVD (SVDAP), ~\ref{alg:AP} and greedy coloring on   undirected random graphs $G(n,p)$ with $p=0.2.$}\label{fig:SVDundirRmin}
 \end{figure}

 %%%%%%%%%%%%%Figure%%%%%%%%%%%%%%%%%%%% 
  \begin{figure}[h!]
\centering \scriptsize
      \input{undirected_SVD_p=0.2_iter=500_Time.tikz}
 \vspace{-0.3cm}
\caption{\footnotesize  Average running time of AP  via eigenvalue decomposition (\ref{alg:AP}) vs.  SVD decomposition (SVDAP) on   undirected random graphs $G(n,p)$ with $p=0.2.$}\label{fig:SVDundirTime}
\vspace{-0.2cm}
 \end{figure}
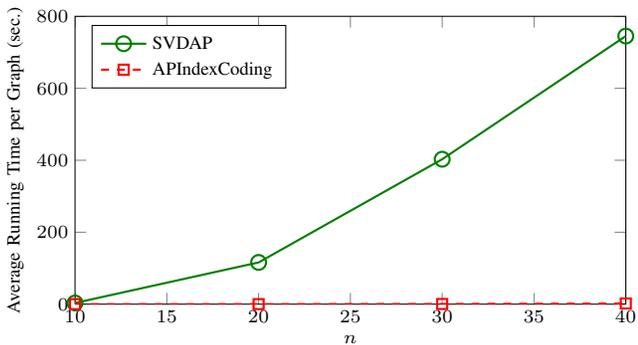

 %%%%%%%%%%%%%%Figure%%%%%%%%%%%%%%%%  
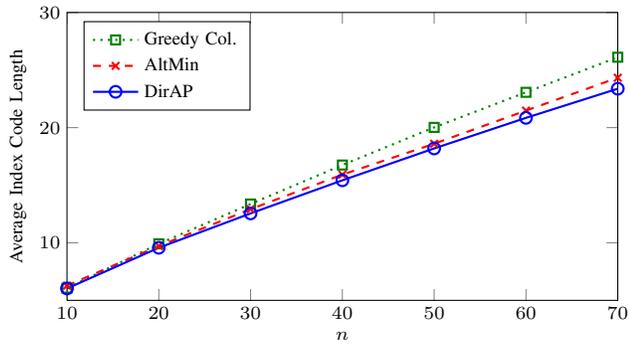
\begin{figure}[h!]\centering\scriptsize
\input{undirected_AM_p=0.2_iter=100_Rmin.tikz}\vspace{-8pt}
\caption{\footnotesize Average index code length obtained by Alternating Minimization and Directional APIndexCoding (DirAP) for undirected random graphs $G(n,p)$ with $p=0.2$.}\label{fig:AMp02}
\vspace{-3pt}
\end{figure}

%%%%%%%%%%%%%%%Figure%%%%%%%%%%%%%%%%
 %\begin{figure}[h!] \centering\scriptsize
%\input{undirected_AM_p=0.2_iter=100_Time.tikz}
%\vspace{-8pt}
%\caption{\footnotesize Average running time of Alternating Minimization and Directional APIndexCoding (DirAP) for undirected random graphs when $p=0.2$.}\label{fig:AMp02Time}
%\vspace{-0.1cm}
%\end{figure}

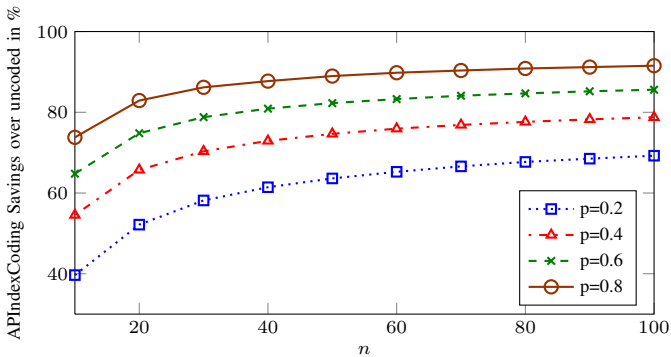
\begin{figure}[h!] \centering\scriptsize
\input{AltPrj_vs_noncoding.tikz}
\vspace{-0.5cm}
\caption{\footnotesize Savings in percentage of APIndexCoding over uncoded transmissions.}\label{fig:APPer2}
\vspace{-0.2cm}
 \end{figure}

 %%%%%%%%%%%%%%Figure%%%%%%%%%%%%%%%%
  \setlength\figureheight{4cm}

\begin{figure}[h!] \centering\scriptsize
\input{AltPrj_vs_GreedyColor_iter=1000.tikz}
\vspace{-0.5cm}
\caption{\footnotesize Savings in percentage of APIndexCoding over greedy coloring.}\label{fig:APPer3}
 \end{figure}
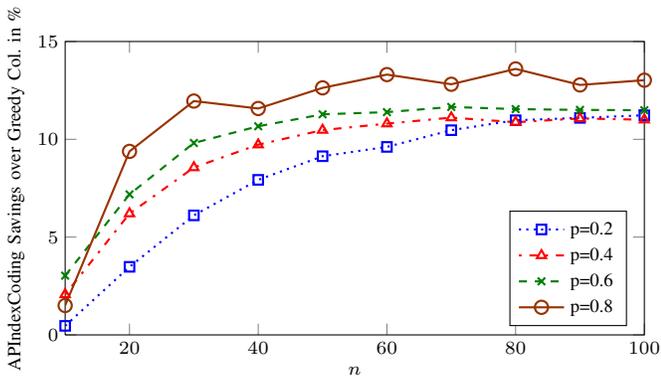
 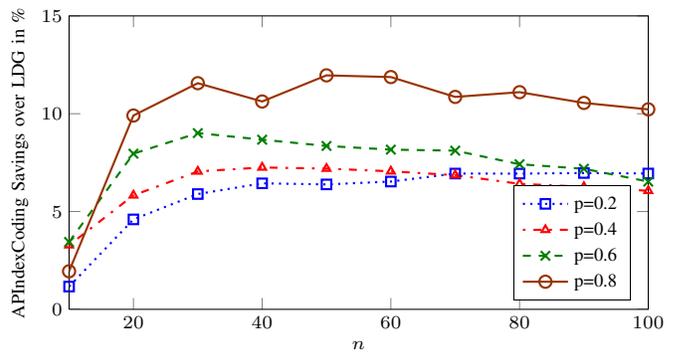
\begin{figure}[h!] \centering\scriptsize
\input{AltPrj_vs_LDG_iter=1000.tikz}
\vspace{-0.45cm}
\caption{\footnotesize Savings in percentage of APIndexCoding over LDG.}\label{fig:APPer4}
\vspace{-0.15cm}
 \end{figure}

 \begin{figure}[h!]
  \centering
\scriptsize
 \input{AltPrj_Time_of_Iteration.tikz}
 \vspace{-0.25cm}
 \caption{\footnotesize Average running time of one iteration by using APIndexCoding on random undirected graphs $G(n,p)$.}
 \vspace{-0.15cm}
\end{figure}
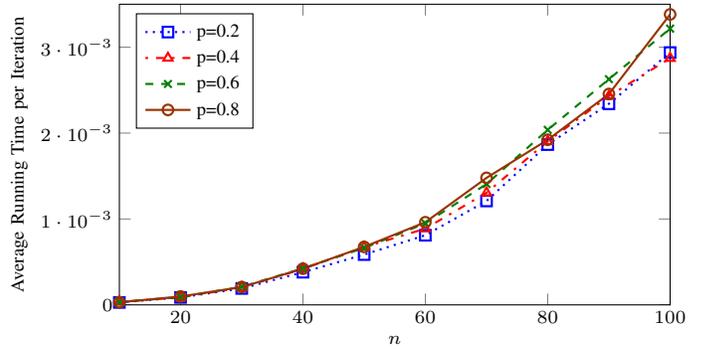

  %%%%%%%%%%%%%Figure%%%%%%%%%%%%%%%%%%%% 
   \setlength\figureheight{3.85cm} 
 \begin{figure}[h!]
  \centering
\scriptsize
 \input{AltPrj_ITE.tikz}
 \vspace{-0.32cm}
 \caption{\footnotesize Average iteration number of one Graph by using APIndexCoding on random undirected graphs $G(n,p)$.} \label{fig:undirectedIt}
 \vspace{-0.15cm}
\end{figure}
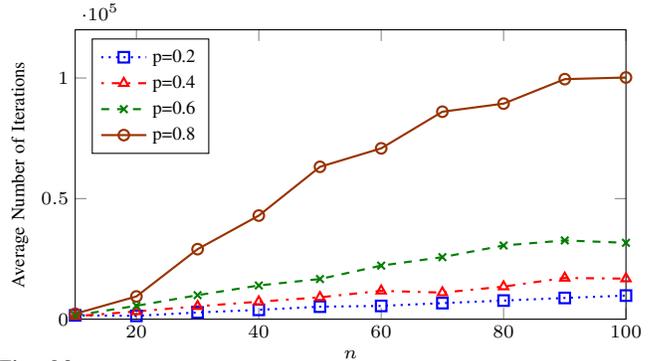

\begin{figure}[h!]
\centering\scriptsize
  \setlength\figureheight{5cm} 
\setlength\figurewidth{9.3cm}
      \input{Directed_cache_AltPrj_Rmin.tikz}\vspace{-0.2cm}
\caption{\footnotesize Average index code length for  directed $c$-regular random graphs for fixed  cache size $c$ fixed. }
 \label{fig:cache}
\end{figure}
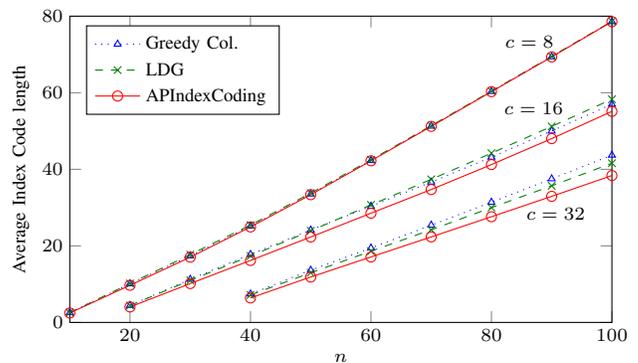

\end{appendices}
\newpage

\bibliographystyle{ieeetr}
\bibliography{coding}

\end{document}

%% file: directed_AP_n=100_Rmin.tikz
% This file was created by matlab2tikz v0.5.0 running on MATLAB 8.1.
% Copyright (c) 2008--2014, Nico Schlömer <nico.schloemer@gmail.com>
% All rights reserved.
% Minimal pgfplots version: 1.3
% 
% The latest updates can be retrieved from
%   http://www.mathworks.com/matlabcentral/fileexchange/22022-matlab2tikz
% where you can also make suggestions and rate matlab2tikz.
% 
%
% defining custom colors
\definecolor{mycolor1}{rgb}{0.60000,0.20000,0.00000}%
\definecolor{mycolor2}{rgb}{0.00000,0.49804,0.00000}%
\begin{tikzpicture}

\begin{axis}[%
width=0.950920245398773\figurewidth,
height=\figureheight,
at={(0\figurewidth,0\figureheight)},
scale only axis,
xmin=0.01,
xmax=1,
xlabel={$p$},
ymin=0,
ymax=100,
ylabel={Average Index Code Length},
legend style={at={(0.008,0.008)},anchor=south west,draw=black,fill=white,legend cell align=left}
]
\addplot [color=mycolor1,dashed,line width=1.6pt,mark size=1.3pt,mark=square*,mark options={solid,fill=mycolor1}]
  table[row sep=crcr]{%
0.01	100\\
0.05	100\\
0.1	100\\
0.2	100\\
0.3	100\\
0.4	100\\
0.5	100\\
0.6	100\\
0.7	100\\
0.8	100\\
0.9	100\\
1	100\\
};
\addlegendentry{No Coding};

\addplot [color=mycolor2,dash pattern=on 1pt off 3pt on 3pt off 3pt,line width=1.0pt,mark=triangle*,mark options={solid,fill=mycolor2}]
  table[row sep=crcr]{%
0.01	100\\
0.05	99.4\\
0.1	96.8\\
0.2	89.54\\
0.3	81.32\\
0.4	72.22\\
0.5	62.89\\
0.6	52.83\\
0.7	42.12\\
0.8	31.45\\
0.9	18.97\\
1	1\\
};
\addlegendentry{Multicast};

\addplot [color=blue,dotted,line width=0.7pt,mark size=2.5pt,mark=x,mark options={solid}]
  table[row sep=crcr]{%
0.01	99.507\\
0.05	89.9\\
0.1	72.84\\
0.2	52.83\\
0.3	45.23\\
0.4	37.97\\
0.5	31.61\\
0.6	25.46\\
0.7	20.18\\
0.8	15\\
0.9	9.36\\
1	1\\
};
\addlegendentry{Greedy Coloring};

\addplot [color=red,solid,line width=0.6pt,mark size=1.0pt,mark=*,mark options={solid}]
  table[row sep=crcr]{%
0.01	99.507\\
0.05	89.9\\
0.1	72.63\\
0.2	50.13\\
0.3	40.76\\
0.4	33.57\\
0.5	27.91\\
0.6	22.78\\
0.7	17.77\\
0.8	13.18\\
0.9	8.11\\
1	1\\
};
\addlegendentry{Alternating Proj.};

\end{axis}
\end{tikzpicture}%

%% file: AltPrj_LDG_GreedyColor_iter=1000.tikz
% This file was created by matlab2tikz v0.5.0 running on MATLAB 8.1.
% Copyright (c) 2008--2014, Nico Schlömer <nico.schloemer@gmail.com>
% All rights reserved.
% Minimal pgfplots version: 1.3
% 
% The latest updates can be retrieved from
%   http://www.mathworks.com/matlabcentral/fileexchange/22022-matlab2tikz
% where you can also make suggestions and rate matlab2tikz.
% 
%
% defining custom colors
\definecolor{mycolor1}{rgb}{0.00000,0.49804,0.00000}%
\definecolor{mycolor2}{rgb}{0.60000,0.20000,0.00000}%
\begin{tikzpicture}

\begin{axis}[%
width=\figurewidth,
height=\figureheight,
at={(0\figurewidth,0\figureheight)},
scale only axis,
xmin=0,
xmax=1,
ymin=0,
ymax=1,
hide axis,
axis x line*=bottom,
axis y line*=left
]
\node[below, right, inner sep=0mm, text=blue]
at (rel axis cs:0.716718490156295,0.88905750299951)(rel axis cs:0.0586745027124774,0.0474934036939302) {$p=0.2$};
\node[below, right, inner sep=0mm, text=red, draw=white]
at (rel axis cs:0.729567522961341,0.648675158780811)(rel axis cs:0.0586745027124774,0.0474934036939302) {$p=0.4$};
\node[below, right, inner sep=0mm, text=mycolor1, draw=white]
at (rel axis cs:0.748861679634149,0.478254393018619)(rel axis cs:0.0586745027124774,0.0474934036939302) {$p=0.6$};
\node[below, right, inner sep=0mm, text=mycolor2, draw=white]
at (rel axis cs:0.764399287901997,0.334026076018144)(rel axis cs:0.0586745027124774,0.0474934036939302) {$p=0.8$};
\end{axis}

\begin{axis}[%
width=0.9546875\figurewidth,
height=0.901408450704225\figureheight,
at={(0.03671875\figurewidth,0.0824949698189135\figureheight)},
scale only axis,
xmin=10,
xmax=100,
xlabel={$n$},
ymin=0,
ymax=35,
ylabel={Average Index Code Length},
legend style={at={(0.03,0.97)},anchor=north west,draw=black,fill=white,legend cell align=left}
]
\addplot [color=blue,solid,line width=0.6pt,forget plot]
  table[row sep=crcr]{%
10	6.033\\
20	9.567\\
30	12.551\\
40	15.425\\
50	18.193\\
60	20.849\\
70	23.38\\
80	25.837\\
90	28.341\\
100	30.758\\
};
\addplot [color=blue,dashed,line width=0.6pt,forget plot]
  table[row sep=crcr]{%
10	6.104\\
20	10.028\\
30	13.337\\
40	16.487\\
50	19.434\\
60	22.306\\
70	25.124\\
80	27.763\\
90	30.464\\
100	33.053\\
};
\addplot [color=blue,dotted,line width=0.6pt,forget plot]
  table[row sep=crcr]{%
10	6.061\\
20	9.912\\
30	13.368\\
40	16.753\\
50	20.022\\
60	23.066\\
70	26.113\\
80	29.022\\
90	31.878\\
100	34.648\\
};
\addplot [color=red,solid,line width=0.6pt,forget plot]
  table[row sep=crcr]{%
10	4.545\\
20	6.852\\
30	8.901\\
40	10.832\\
50	12.67\\
60	14.437\\
70	16.188\\
80	17.9\\
90	19.577\\
100	21.244\\
};
\addplot [color=red,dashed,line width=0.6pt,forget plot]
  table[row sep=crcr]{%
10	4.699\\
20	7.276\\
30	9.576\\
40	11.679\\
50	13.652\\
60	15.533\\
70	17.38\\
80	19.127\\
90	20.885\\
100	22.614\\
};
\addplot [color=red,dotted,line width=0.6pt,forget plot]
  table[row sep=crcr]{%
10	4.641\\
20	7.305\\
30	9.734\\
40	11.999\\
50	14.152\\
60	16.185\\
70	18.211\\
80	20.084\\
90	22.013\\
100	23.87\\
};
\addplot [color=mycolor1,solid,line width=0.6pt,forget plot]
  table[row sep=crcr]{%
10	3.52111111111111\\
20	5.03777777777778\\
30	6.36222222222222\\
40	7.64\\
50	8.86333333333333\\
60	10.0366666666667\\
70	11.1266666666667\\
80	12.26\\
90	13.3177777777778\\
100	14.4088888888889\\
};
\addplot [color=mycolor1,dashed,line width=0.6pt,forget plot]
  table[row sep=crcr]{%
10	3.64666666666667\\
20	5.47333333333333\\
30	6.99222222222222\\
40	8.36555555555556\\
50	9.67111111111111\\
60	10.9288888888889\\
70	12.1088888888889\\
80	13.2422222222222\\
90	14.3488888888889\\
100	15.4155555555556\\
};
\addplot [color=black!50!green,dotted,line width=0.6pt,forget plot]
  table[row sep=crcr]{%
10	3.63111111111111\\
20	5.42777777777778\\
30	7.05444444444444\\
40	8.55222222222222\\
50	9.99\\
60	11.3266666666667\\
70	12.5944444444444\\
80	13.86\\
90	15.0488888888889\\
100	16.2788888888889\\
};
\addplot [color=mycolor2,solid,line width=0.6pt,forget plot]
  table[row sep=crcr]{%
10	2.621\\
20	3.419\\
30	4.146\\
40	4.912\\
50	5.513\\
60	6.123\\
70	6.756\\
80	7.319\\
90	7.9284064665127\\
100	8.46\\
};
\addplot [color=mycolor2,dashed,line width=0.6pt,forget plot]
  table[row sep=crcr]{%
10	2.673\\
20	3.795\\
30	4.688\\
40	5.496\\
50	6.262\\
60	6.948\\
70	7.579\\
80	8.233\\
90	8.86374133949192\\
100	9.42333333333333\\
};
\addplot [color=mycolor2,dotted,line width=0.6pt,forget plot]
  table[row sep=crcr]{%
10	2.661\\
20	3.773\\
30	4.709\\
40	5.555\\
50	6.31\\
60	7.063\\
70	7.749\\
80	8.471\\
90	9.09006928406466\\
100	9.72666666666667\\
};
\addplot [color=black,solid,line width=0.6pt]
  table[row sep=crcr]{%
10	2.673\\
};
\addlegendentry{APIndexCoding};

\addplot [color=black,dashed,line width=0.6pt]
  table[row sep=crcr]{%
10	2.673\\
};
\addlegendentry{LDG};

\addplot [color=black,dotted,line width=0.6pt]
  table[row sep=crcr]{%
10	2.661\\
};
\addlegendentry{Greedy Col.};

\end{axis}
\end{tikzpicture}%

%% file: AltPrj_pdf_n=100_iter=300.tikz
% This file was created by matlab2tikz v0.5.0 running on MATLAB 8.1.
% Copyright (c) 2008--2014, Nico Schlömer <nico.schloemer@gmail.com>
% All rights reserved.
% Minimal pgfplots version: 1.3
% 
% The latest updates can be retrieved from
%   http://www.mathworks.com/matlabcentral/fileexchange/22022-matlab2tikz
% where you can also make suggestions and rate matlab2tikz.
% 
%
% defining custom colors
\definecolor{mycolor1}{rgb}{0.00000,0.49804,0.00000}%
\definecolor{mycolor2}{rgb}{0.60000,0.20000,0.00000}%
\begin{tikzpicture}

\begin{axis}[%
width=1\figurewidth,
height=\figureheight,
at={(0\figurewidth,0\figureheight)},
scale only axis,
xmin=0,
xmax=1,
ymin=0,
ymax=1,
hide axis,
axis x line*=bottom,
axis y line*=left
]
\node[below, right, inner sep=0mm, text=black]
at (rel axis cs:0.224009433999579,0.698463907338326)(rel axis cs:0.046870927174236,0.0438548812664901) {$p=0.4$};
\node[below, right, inner sep=0mm, text=black]
at (rel axis cs:0.319079756591063,0.551725666959833)(rel axis cs:0.0511234036798173,0.0438548812664901) {$p=0.2$};
\node[below, right, inner sep=0mm, text=black]
at (rel axis cs:0.155678612667669,0.776947249770013)(rel axis cs:0.0463529900411342,0.0438548812664901) {$p=0.6$};
\node[below, right, inner sep=0mm, text=black]
at (rel axis cs:0.1057011860311,0.940975027547791)(rel axis cs:0.0463529900411342,0.0438548812664901) {$p=0.8$};
\end{axis}

\begin{axis}[%
width=0.939736346516008\figurewidth,
height=0.887201735357918\figureheight,
at={(0.0489642184557439\figurewidth,0.0954446854663774\figureheight)},
scale only axis,
xmin=0,
xmax=100,
xlabel={Index Code Length},
ymin=0,
ymax=0.7,
ylabel={Histogram},
legend style={draw=black,fill=white,legend cell align=left}
]
\addplot [color=blue,solid,line width=0.6pt,forget plot]
  table[row sep=crcr]{%
1	0\\
2	0\\
3	0\\
4	0\\
5	0\\
6	0\\
7	0\\
8	0\\
9	0\\
10	0\\
11	0\\
12	0\\
13	0\\
14	0\\
15	0\\
16	0\\
17	0\\
18	0\\
19	0\\
20	0\\
21	0\\
22	0\\
23	0\\
24	0\\
25	0\\
26	0\\
27	0\\
28	0.01\\
29	0.09\\
30	0.33\\
31	0.34\\
32	0.166666666666667\\
33	0.0533333333333333\\
34	0.00666666666666667\\
35	0.00333333333333333\\
36	0\\
37	0\\
38	0\\
39	0\\
40	0\\
41	0\\
42	0\\
43	0\\
44	0\\
45	0\\
46	0\\
47	0\\
48	0\\
49	0\\
50	0\\
51	0\\
52	0\\
53	0\\
54	0\\
55	0\\
56	0\\
57	0\\
58	0\\
59	0\\
60	0\\
61	0\\
62	0\\
63	0\\
64	0\\
65	0\\
66	0\\
67	0\\
68	0\\
69	0\\
70	0\\
71	0\\
72	0\\
73	0\\
74	0\\
75	0\\
76	0\\
77	0\\
78	0\\
79	0\\
80	0\\
81	0\\
82	0\\
83	0\\
84	0\\
85	0\\
86	0\\
87	0\\
88	0\\
89	0\\
90	0\\
91	0\\
92	0\\
93	0\\
94	0\\
95	0\\
96	0\\
97	0\\
98	0\\
99	0\\
100	0\\
};
\addplot [color=blue,dashed,line width=0.6pt,forget plot]
  table[row sep=crcr]{%
1	0\\
2	0\\
3	0\\
4	0\\
5	0\\
6	0\\
7	0\\
8	0\\
9	0\\
10	0\\
11	0\\
12	0\\
13	0\\
14	0\\
15	0\\
16	0\\
17	0\\
18	0\\
19	0\\
20	0\\
21	0\\
22	0\\
23	0\\
24	0\\
25	0\\
26	0\\
27	0\\
28	0\\
29	0\\
30	0.0166666666666667\\
31	0.0933333333333333\\
32	0.186666666666667\\
33	0.346666666666667\\
34	0.246666666666667\\
35	0.1\\
36	0.01\\
37	0\\
38	0\\
39	0\\
40	0\\
41	0\\
42	0\\
43	0\\
44	0\\
45	0\\
46	0\\
47	0\\
48	0\\
49	0\\
50	0\\
51	0\\
52	0\\
53	0\\
54	0\\
55	0\\
56	0\\
57	0\\
58	0\\
59	0\\
60	0\\
61	0\\
62	0\\
63	0\\
64	0\\
65	0\\
66	0\\
67	0\\
68	0\\
69	0\\
70	0\\
71	0\\
72	0\\
73	0\\
74	0\\
75	0\\
76	0\\
77	0\\
78	0\\
79	0\\
80	0\\
81	0\\
82	0\\
83	0\\
84	0\\
85	0\\
86	0\\
87	0\\
88	0\\
89	0\\
90	0\\
91	0\\
92	0\\
93	0\\
94	0\\
95	0\\
96	0\\
97	0\\
98	0\\
99	0\\
100	0\\
};
\addplot [color=blue,dotted,line width=0.6pt,forget plot]
  table[row sep=crcr]{%
1	0\\
2	0\\
3	0\\
4	0\\
5	0\\
6	0\\
7	0\\
8	0\\
9	0\\
10	0\\
11	0\\
12	0\\
13	0\\
14	0\\
15	0\\
16	0\\
17	0\\
18	0\\
19	0\\
20	0\\
21	0\\
22	0\\
23	0\\
24	0\\
25	0\\
26	0\\
27	0\\
28	0\\
29	0\\
30	0\\
31	0\\
32	0.04\\
33	0.0966666666666667\\
34	0.293333333333333\\
35	0.316666666666667\\
36	0.17\\
37	0.0633333333333333\\
38	0.02\\
39	0\\
40	0\\
41	0\\
42	0\\
43	0\\
44	0\\
45	0\\
46	0\\
47	0\\
48	0\\
49	0\\
50	0\\
51	0\\
52	0\\
53	0\\
54	0\\
55	0\\
56	0\\
57	0\\
58	0\\
59	0\\
60	0\\
61	0\\
62	0\\
63	0\\
64	0\\
65	0\\
66	0\\
67	0\\
68	0\\
69	0\\
70	0\\
71	0\\
72	0\\
73	0\\
74	0\\
75	0\\
76	0\\
77	0\\
78	0\\
79	0\\
80	0\\
81	0\\
82	0\\
83	0\\
84	0\\
85	0\\
86	0\\
87	0\\
88	0\\
89	0\\
90	0\\
91	0\\
92	0\\
93	0\\
94	0\\
95	0\\
96	0\\
97	0\\
98	0\\
99	0\\
100	0\\
};
\addplot [color=red,solid,line width=0.6pt,forget plot]
  table[row sep=crcr]{%
1	0\\
2	0\\
3	0\\
4	0\\
5	0\\
6	0\\
7	0\\
8	0\\
9	0\\
10	0\\
11	0\\
12	0\\
13	0\\
14	0\\
15	0\\
16	0\\
17	0\\
18	0\\
19	0.018\\
20	0.169\\
21	0.459\\
22	0.276\\
23	0.064\\
24	0.011\\
25	0.003\\
26	0\\
27	0\\
28	0\\
29	0\\
30	0\\
31	0\\
32	0\\
33	0\\
34	0\\
35	0\\
36	0\\
37	0\\
38	0\\
39	0\\
40	0\\
41	0\\
42	0\\
43	0\\
44	0\\
45	0\\
46	0\\
47	0\\
48	0\\
49	0\\
50	0\\
51	0\\
52	0\\
53	0\\
54	0\\
55	0\\
56	0\\
57	0\\
58	0\\
59	0\\
60	0\\
61	0\\
62	0\\
63	0\\
64	0\\
65	0\\
66	0\\
67	0\\
68	0\\
69	0\\
70	0\\
71	0\\
72	0\\
73	0\\
74	0\\
75	0\\
76	0\\
77	0\\
78	0\\
79	0\\
80	0\\
81	0\\
82	0\\
83	0\\
84	0\\
85	0\\
86	0\\
87	0\\
88	0\\
89	0\\
90	0\\
91	0\\
92	0\\
93	0\\
94	0\\
95	0\\
96	0\\
97	0\\
98	0\\
99	0\\
100	0\\
};
\addplot [color=red,dashed,line width=0.6pt,forget plot]
  table[row sep=crcr]{%
1	0\\
2	0\\
3	0\\
4	0\\
5	0\\
6	0\\
7	0\\
8	0\\
9	0\\
10	0\\
11	0\\
12	0\\
13	0\\
14	0\\
15	0\\
16	0\\
17	0\\
18	0\\
19	0\\
20	0.003\\
21	0.091\\
22	0.365\\
23	0.385\\
24	0.142\\
25	0.014\\
26	0\\
27	0\\
28	0\\
29	0\\
30	0\\
31	0\\
32	0\\
33	0\\
34	0\\
35	0\\
36	0\\
37	0\\
38	0\\
39	0\\
40	0\\
41	0\\
42	0\\
43	0\\
44	0\\
45	0\\
46	0\\
47	0\\
48	0\\
49	0\\
50	0\\
51	0\\
52	0\\
53	0\\
54	0\\
55	0\\
56	0\\
57	0\\
58	0\\
59	0\\
60	0\\
61	0\\
62	0\\
63	0\\
64	0\\
65	0\\
66	0\\
67	0\\
68	0\\
69	0\\
70	0\\
71	0\\
72	0\\
73	0\\
74	0\\
75	0\\
76	0\\
77	0\\
78	0\\
79	0\\
80	0\\
81	0\\
82	0\\
83	0\\
84	0\\
85	0\\
86	0\\
87	0\\
88	0\\
89	0\\
90	0\\
91	0\\
92	0\\
93	0\\
94	0\\
95	0\\
96	0\\
97	0\\
98	0\\
99	0\\
100	0\\
};
\addplot [color=red,dotted,line width=0.6pt,forget plot]
  table[row sep=crcr]{%
1	0\\
2	0\\
3	0\\
4	0\\
5	0\\
6	0\\
7	0\\
8	0\\
9	0\\
10	0\\
11	0\\
12	0\\
13	0\\
14	0\\
15	0\\
16	0\\
17	0\\
18	0\\
19	0\\
20	0\\
21	0.005\\
22	0.071\\
23	0.264\\
24	0.415\\
25	0.202\\
26	0.04\\
27	0.003\\
28	0\\
29	0\\
30	0\\
31	0\\
32	0\\
33	0\\
34	0\\
35	0\\
36	0\\
37	0\\
38	0\\
39	0\\
40	0\\
41	0\\
42	0\\
43	0\\
44	0\\
45	0\\
46	0\\
47	0\\
48	0\\
49	0\\
50	0\\
51	0\\
52	0\\
53	0\\
54	0\\
55	0\\
56	0\\
57	0\\
58	0\\
59	0\\
60	0\\
61	0\\
62	0\\
63	0\\
64	0\\
65	0\\
66	0\\
67	0\\
68	0\\
69	0\\
70	0\\
71	0\\
72	0\\
73	0\\
74	0\\
75	0\\
76	0\\
77	0\\
78	0\\
79	0\\
80	0\\
81	0\\
82	0\\
83	0\\
84	0\\
85	0\\
86	0\\
87	0\\
88	0\\
89	0\\
90	0\\
91	0\\
92	0\\
93	0\\
94	0\\
95	0\\
96	0\\
97	0\\
98	0\\
99	0\\
100	0\\
};
\addplot [color=mycolor1,solid,line width=0.6pt,forget plot]
  table[row sep=crcr]{%
1	0\\
2	0\\
3	0\\
4	0\\
5	0\\
6	0\\
7	0\\
8	0\\
9	0\\
10	0\\
11	0\\
12	0\\
13	0.0744444444444444\\
14	0.523333333333333\\
15	0.331111111111111\\
16	0.0611111111111111\\
17	0.01\\
18	0\\
19	0\\
20	0\\
21	0\\
22	0\\
23	0\\
24	0\\
25	0\\
26	0\\
27	0\\
28	0\\
29	0\\
30	0\\
31	0\\
32	0\\
33	0\\
34	0\\
35	0\\
36	0\\
37	0\\
38	0\\
39	0\\
40	0\\
41	0\\
42	0\\
43	0\\
44	0\\
45	0\\
46	0\\
47	0\\
48	0\\
49	0\\
50	0\\
51	0\\
52	0\\
53	0\\
54	0\\
55	0\\
56	0\\
57	0\\
58	0\\
59	0\\
60	0\\
61	0\\
62	0\\
63	0\\
64	0\\
65	0\\
66	0\\
67	0\\
68	0\\
69	0\\
70	0\\
71	0\\
72	0\\
73	0\\
74	0\\
75	0\\
76	0\\
77	0\\
78	0\\
79	0\\
80	0\\
81	0\\
82	0\\
83	0\\
84	0\\
85	0\\
86	0\\
87	0\\
88	0\\
89	0\\
90	0\\
91	0\\
92	0\\
93	0\\
94	0\\
95	0\\
96	0\\
97	0\\
98	0\\
99	0\\
100	0\\
};
\addplot [color=mycolor1,dashed,line width=0.6pt,forget plot]
  table[row sep=crcr]{%
1	0\\
2	0\\
3	0\\
4	0\\
5	0\\
6	0\\
7	0\\
8	0\\
9	0\\
10	0\\
11	0\\
12	0\\
13	0\\
14	0.0844444444444444\\
15	0.482222222222222\\
16	0.368888888888889\\
17	0.0622222222222222\\
18	0.00222222222222222\\
19	0\\
20	0\\
21	0\\
22	0\\
23	0\\
24	0\\
25	0\\
26	0\\
27	0\\
28	0\\
29	0\\
30	0\\
31	0\\
32	0\\
33	0\\
34	0\\
35	0\\
36	0\\
37	0\\
38	0\\
39	0\\
40	0\\
41	0\\
42	0\\
43	0\\
44	0\\
45	0\\
46	0\\
47	0\\
48	0\\
49	0\\
50	0\\
51	0\\
52	0\\
53	0\\
54	0\\
55	0\\
56	0\\
57	0\\
58	0\\
59	0\\
60	0\\
61	0\\
62	0\\
63	0\\
64	0\\
65	0\\
66	0\\
67	0\\
68	0\\
69	0\\
70	0\\
71	0\\
72	0\\
73	0\\
74	0\\
75	0\\
76	0\\
77	0\\
78	0\\
79	0\\
80	0\\
81	0\\
82	0\\
83	0\\
84	0\\
85	0\\
86	0\\
87	0\\
88	0\\
89	0\\
90	0\\
91	0\\
92	0\\
93	0\\
94	0\\
95	0\\
96	0\\
97	0\\
98	0\\
99	0\\
100	0\\
};
\addplot [color=black!50!green,dotted,line width=0.6pt,forget plot]
  table[row sep=crcr]{%
1	0\\
2	0\\
3	0\\
4	0\\
5	0\\
6	0\\
7	0\\
8	0\\
9	0\\
10	0\\
11	0\\
12	0\\
13	0\\
14	0.00222222222222222\\
15	0.128888888888889\\
16	0.496666666666667\\
17	0.332222222222222\\
18	0.04\\
19	0\\
20	0\\
21	0\\
22	0\\
23	0\\
24	0\\
25	0\\
26	0\\
27	0\\
28	0\\
29	0\\
30	0\\
31	0\\
32	0\\
33	0\\
34	0\\
35	0\\
36	0\\
37	0\\
38	0\\
39	0\\
40	0\\
41	0\\
42	0\\
43	0\\
44	0\\
45	0\\
46	0\\
47	0\\
48	0\\
49	0\\
50	0\\
51	0\\
52	0\\
53	0\\
54	0\\
55	0\\
56	0\\
57	0\\
58	0\\
59	0\\
60	0\\
61	0\\
62	0\\
63	0\\
64	0\\
65	0\\
66	0\\
67	0\\
68	0\\
69	0\\
70	0\\
71	0\\
72	0\\
73	0\\
74	0\\
75	0\\
76	0\\
77	0\\
78	0\\
79	0\\
80	0\\
81	0\\
82	0\\
83	0\\
84	0\\
85	0\\
86	0\\
87	0\\
88	0\\
89	0\\
90	0\\
91	0\\
92	0\\
93	0\\
94	0\\
95	0\\
96	0\\
97	0\\
98	0\\
99	0\\
100	0\\
};
\addplot [color=mycolor2,solid,line width=0.6pt,forget plot]
  table[row sep=crcr]{%
1	0\\
2	0\\
3	0\\
4	0\\
5	0\\
6	0\\
7	0.00666666666666667\\
8	0.55\\
9	0.423333333333333\\
10	0.0166666666666667\\
11	0.00333333333333333\\
12	0\\
13	0\\
14	0\\
15	0\\
16	0\\
17	0\\
18	0\\
19	0\\
20	0\\
21	0\\
22	0\\
23	0\\
24	0\\
25	0\\
26	0\\
27	0\\
28	0\\
29	0\\
30	0\\
31	0\\
32	0\\
33	0\\
34	0\\
35	0\\
36	0\\
37	0\\
38	0\\
39	0\\
40	0\\
41	0\\
42	0\\
43	0\\
44	0\\
45	0\\
46	0\\
47	0\\
48	0\\
49	0\\
50	0\\
51	0\\
52	0\\
53	0\\
54	0\\
55	0\\
56	0\\
57	0\\
58	0\\
59	0\\
60	0\\
61	0\\
62	0\\
63	0\\
64	0\\
65	0\\
66	0\\
67	0\\
68	0\\
69	0\\
70	0\\
71	0\\
72	0\\
73	0\\
74	0\\
75	0\\
76	0\\
77	0\\
78	0\\
79	0\\
80	0\\
81	0\\
82	0\\
83	0\\
84	0\\
85	0\\
86	0\\
87	0\\
88	0\\
89	0\\
90	0\\
91	0\\
92	0\\
93	0\\
94	0\\
95	0\\
96	0\\
97	0\\
98	0\\
99	0\\
100	0\\
};
\addplot [color=mycolor2,dashed,line width=0.6pt,forget plot]
  table[row sep=crcr]{%
1	0\\
2	0\\
3	0\\
4	0\\
5	0\\
6	0\\
7	0\\
8	0.02\\
9	0.55\\
10	0.416666666666667\\
11	0.0133333333333333\\
12	0\\
13	0\\
14	0\\
15	0\\
16	0\\
17	0\\
18	0\\
19	0\\
20	0\\
21	0\\
22	0\\
23	0\\
24	0\\
25	0\\
26	0\\
27	0\\
28	0\\
29	0\\
30	0\\
31	0\\
32	0\\
33	0\\
34	0\\
35	0\\
36	0\\
37	0\\
38	0\\
39	0\\
40	0\\
41	0\\
42	0\\
43	0\\
44	0\\
45	0\\
46	0\\
47	0\\
48	0\\
49	0\\
50	0\\
51	0\\
52	0\\
53	0\\
54	0\\
55	0\\
56	0\\
57	0\\
58	0\\
59	0\\
60	0\\
61	0\\
62	0\\
63	0\\
64	0\\
65	0\\
66	0\\
67	0\\
68	0\\
69	0\\
70	0\\
71	0\\
72	0\\
73	0\\
74	0\\
75	0\\
76	0\\
77	0\\
78	0\\
79	0\\
80	0\\
81	0\\
82	0\\
83	0\\
84	0\\
85	0\\
86	0\\
87	0\\
88	0\\
89	0\\
90	0\\
91	0\\
92	0\\
93	0\\
94	0\\
95	0\\
96	0\\
97	0\\
98	0\\
99	0\\
100	0\\
};
\addplot [color=mycolor2,dotted,line width=0.6pt,forget plot]
  table[row sep=crcr]{%
1	0\\
2	0\\
3	0\\
4	0\\
5	0\\
6	0\\
7	0\\
8	0.00666666666666667\\
9	0.293333333333333\\
10	0.67\\
11	0.0266666666666667\\
12	0.00333333333333333\\
13	0\\
14	0\\
15	0\\
16	0\\
17	0\\
18	0\\
19	0\\
20	0\\
21	0\\
22	0\\
23	0\\
24	0\\
25	0\\
26	0\\
27	0\\
28	0\\
29	0\\
30	0\\
31	0\\
32	0\\
33	0\\
34	0\\
35	0\\
36	0\\
37	0\\
38	0\\
39	0\\
40	0\\
41	0\\
42	0\\
43	0\\
44	0\\
45	0\\
46	0\\
47	0\\
48	0\\
49	0\\
50	0\\
51	0\\
52	0\\
53	0\\
54	0\\
55	0\\
56	0\\
57	0\\
58	0\\
59	0\\
60	0\\
61	0\\
62	0\\
63	0\\
64	0\\
65	0\\
66	0\\
67	0\\
68	0\\
69	0\\
70	0\\
71	0\\
72	0\\
73	0\\
74	0\\
75	0\\
76	0\\
77	0\\
78	0\\
79	0\\
80	0\\
81	0\\
82	0\\
83	0\\
84	0\\
85	0\\
86	0\\
87	0\\
88	0\\
89	0\\
90	0\\
91	0\\
92	0\\
93	0\\
94	0\\
95	0\\
96	0\\
97	0\\
98	0\\
99	0\\
100	0\\
};
\addplot [color=black,solid,line width=0.6pt]
  table[row sep=crcr]{%
0	0\\
};
\addlegendentry{APIndexCoding};

\addplot [color=black,dashed,line width=0.6pt]
  table[row sep=crcr]{%
0	0\\
};
\addlegendentry{LDG};

\addplot [color=black,dotted,line width=0.6pt]
  table[row sep=crcr]{%
0	0\\
};
\addlegendentry{Greedy Col.};

\end{axis}
\end{tikzpicture}%

%% file: AltPrj_vs_Multicast_iter=1000.tikz
% This file was created by matlab2tikz v0.5.0 running on MATLAB 8.1.
% Copyright (c) 2008--2014, Nico Schlömer <nico.schloemer@gmail.com>
% All rights reserved.
% Minimal pgfplots version: 1.3
% 
% The latest updates can be retrieved from
%   http://www.mathworks.com/matlabcentral/fileexchange/22022-matlab2tikz
% where you can also make suggestions and rate matlab2tikz.
% 
%
% defining custom colors
\definecolor{mycolor1}{rgb}{0.00000,0.49804,0.00000}%
\definecolor{mycolor2}{rgb}{0.60000,0.20000,0.00000}%
\begin{tikzpicture}

\begin{axis}[%
width=\figurewidth,
height=0.954033314212116\figureheight,
at={(0\figurewidth,0\figureheight)},
scale only axis,
xmin=10,
xmax=100,
xlabel={$n$},
ymin=35,
ymax=75,
ylabel={APIndexCoding Savings over Multicast in \%},
legend style={at={(0.97,0.03)},anchor=south east,draw=black,fill=white,legend cell align=left}
]
\addplot [color=blue,dotted,line width=0.8pt,mark size=1.8pt,mark=square,mark options={solid}]
  table[row sep=crcr]{%
10	37.8156630728318\\
20	49.7090410181199\\
30	55.3438957656577\\
40	58.3594289910159\\
50	60.3811838388854\\
60	61.8975735722784\\
70	63.1571399981405\\
80	64.1959363766626\\
90	64.9530890029048\\
100	65.6385489945047\\
};
\addlegendentry{p=0.2};

\addplot [color=red,dash pattern=on 1pt off 3pt on 3pt off 3pt,line width=0.8pt,mark=triangle,mark options={solid}]
  table[row sep=crcr]{%
10	46.1199231808806\\
20	57.5306805503905\\
30	62.1551295297134\\
40	64.7517458168731\\
50	66.5153204962181\\
60	67.8168657724554\\
70	68.733039103016\\
80	69.5011824681211\\
90	70.1192513099634\\
100	70.6182152128603\\
};
\addlegendentry{p=0.4};

\addplot [color=mycolor1,dashed,line width=0.8pt,mark size=2.5pt,mark=x,mark options={solid}]
  table[row sep=crcr]{%
10	47.6803698200429\\
20	59.7502634342595\\
30	64.5363056938243\\
40	67.0607306999163\\
50	68.6789194639489\\
60	69.8820187349648\\
70	70.957598790277\\
80	71.6424918524207\\
90	72.2965311263011\\
100	72.7602225310013\\
};
\addlegendentry{p=0.6};

\addplot [color=mycolor2,solid,line width=0.8pt,mark=o,mark options={solid}]
  table[row sep=crcr]{%
10	43.1613644742264\\
20	58.0933003211335\\
30	63.508981930521\\
40	65.931947122387\\
50	68.2867480829963\\
60	69.7584827381834\\
70	70.6703364923355\\
80	71.6580377092538\\
90	72.1748791258675\\
100	72.873834472675\\
};
\addlegendentry{p=0.8};

\end{axis}
\end{tikzpicture}%

%% file: 3color_p=0.5_iter=5000_Rmin.tikz
% This file was created by matlab2tikz v0.5.0 running on MATLAB 8.1.
% Copyright (c) 2008--2014, Nico Schlömer <nico.schloemer@gmail.com>
% All rights reserved.
% Minimal pgfplots version: 1.3
% 
% The latest updates can be retrieved from
%   http://www.mathworks.com/matlabcentral/fileexchange/22022-matlab2tikz
% where you can also make suggestions and rate matlab2tikz.
% 
%
% defining custom colors
\definecolor{mycolor1}{rgb}{0.00000,0.49804,0.00000}%
\begin{tikzpicture}

\begin{axis}[%
width=0.950920245398773\figurewidth,
height=\figureheight,
at={(0\figurewidth,0\figureheight)},
scale only axis,
xmin=10,
xmax=50,
xlabel={$n$},
ymin=3,
ymax=4.4,
ylabel={Average Index Code Length},
legend style={at={(0.03,0.97)},anchor=north west,draw=black,fill=white,legend cell align=left}
]
\addplot [color=blue,dotted,line width=0.9pt,mark size=1.7pt,mark=triangle,mark options={solid}]
  table[row sep=crcr]{%
10	3.1174\\
20	3.7472\\
30	4.0484\\
40	4.156\\
50	4.176\\
};
\addlegendentry{Greedy Col.};

\addplot [color=mycolor1,dashed,line width=0.9pt,mark=x,mark options={solid}]
  table[row sep=crcr]{%
10	3.154\\
20	3.3206\\
30	3.0852\\
40	3.0248\\
50	3.008\\
};
\addlegendentry{LDG};

\addplot [color=red,solid,line width=0.9pt,mark size=2.5pt,mark=o,mark options={solid}]
  table[row sep=crcr]{%
10	3.066\\
20	3.0274\\
30	3.0004\\
40	3\\
50	3\\
};
\addlegendentry{APIndexCoding};

\end{axis}
\end{tikzpicture}%

%% file: AltPrj_Time_consumption_linear.tikz
% This file was created by matlab2tikz v0.5.0 running on MATLAB 8.1.
% Copyright (c) 2008--2014, Nico Schlömer <nico.schloemer@gmail.com>
% All rights reserved.
% Minimal pgfplots version: 1.3
% 
% The latest updates can be retrieved from
%   http://www.mathworks.com/matlabcentral/fileexchange/22022-matlab2tikz
% where you can also make suggestions and rate matlab2tikz.
% 
%
% defining custom colors
\definecolor{mycolor1}{rgb}{0.00000,0.49804,0.00000}%
\definecolor{mycolor2}{rgb}{0.60000,0.20000,0.00000}%
\begin{tikzpicture}

\begin{axis}[%
width=\figurewidth,
height=0.964807817760712\figureheight,
at={(0\figurewidth,0\figureheight)},
scale only axis,
xmin=10,
xmax=100,
xlabel={$n$},
ymin=0,
ymax=350,
ylabel={Average Running Time per Graph (sec.)},
legend style={at={(0.03,0.97)},anchor=north west,draw=black,fill=white,legend cell align=left}
]
\addplot [color=blue,dotted,line width=0.8pt,mark size=1.8pt,mark=square,mark options={solid}]
  table[row sep=crcr]{%
10	0.0443161864575908\\
20	0.118643087954383\\
30	0.530772029935852\\
40	1.46724132031361\\
50	3.03131674354954\\
60	4.4479933274412\\
70	8.0478871147541\\
80	14.3840262637206\\
90	20.5869972399144\\
100	28.7048251840342\\
};
\addlegendentry{p=0.2};

\addplot [color=red,dash pattern=on 1pt off 3pt on 3pt off 3pt,line width=0.8pt,mark=triangle,mark options={solid}]
  table[row sep=crcr]{%
10	0.0407006812373163\\
20	0.281260036086878\\
30	1.08152826589527\\
40	3.04615776907712\\
50	6.07865285939204\\
60	10.4249719943232\\
70	14.3670755924717\\
80	25.6825413930128\\
90	41.5176082836806\\
100	48.2423004597705\\
};
\addlegendentry{p=0.4};

\addplot [color=mycolor1,dashed,line width=0.8pt,mark size=2.5pt,mark=x,mark options={solid}]
  table[row sep=crcr]{%
10	0.0503061391186886\\
20	0.500579941710686\\
30	2.09021897440951\\
40	5.9375591245826\\
50	11.0361402502247\\
60	21.0819453560564\\
70	36.1629764939848\\
80	62.2334917715119\\
90	85.7915666636794\\
100	101.811917298586\\
};
\addlegendentry{p=0.6};

\addplot [color=mycolor2,solid,line width=0.8pt,mark=o,mark options={solid}]
  table[row sep=crcr]{%
10	0.0690612781064212\\
20	0.917043740058851\\
30	6.0233697801037\\
40	18.1001825114461\\
50	42.6577706847893\\
60	68.1379837389241\\
70	127.196991409556\\
80	171.803308887636\\
90	244.334069318762\\
100	338.933046956685\\
};
\addlegendentry{p=0.8};

\end{axis}
\end{tikzpicture}%

%% file: dirAP_p=0.2_iter=1000_Time.tikz
% This file was created by matlab2tikz v0.5.0 running on MATLAB 8.1.
% Copyright (c) 2008--2014, Nico Schlömer <nico.schloemer@gmail.com>
% All rights reserved.
% Minimal pgfplots version: 1.3
% 
% The latest updates can be retrieved from
%   http://www.mathworks.com/matlabcentral/fileexchange/22022-matlab2tikz
% where you can also make suggestions and rate matlab2tikz.
% 
%
% defining custom colors
\definecolor{mycolor1}{rgb}{0.00000,0.49804,0.00000}%
\definecolor{mycolor2}{rgb}{0.60000,0.20000,0.00000}%
\begin{tikzpicture}

\begin{axis}[%
width=0.950920245398773\figurewidth,
height=\figureheight,
at={(0\figurewidth,0\figureheight)},
scale only axis,
xmin=0,
xmax=140,
xlabel={$n$},
ymin=0,
ymax=90,
ylabel={Average Running Time per Graph (sec.)},
legend style={at={(0.03,0.97)},anchor=north west,draw=black,fill=white,legend cell align=left}
]
\addplot [color=blue,dotted,line width=0.8pt,mark size=1.8pt,mark=square,mark options={solid}]
  table[row sep=crcr]{%
10	0.00258732484434171\\
30	0.00633730182239717\\
50	0.00970344915771043\\
70	0.0132767487483332\\
90	0.0167368064015546\\
110	0.0206215868577341\\
130	0.0245577891292516\\
};
\addlegendentry{Greedy Col.};

\addplot [color=red,dash pattern=on 1pt off 3pt on 3pt off 3pt,line width=0.8pt,mark size=1.7pt,mark=triangle,mark options={solid}]
  table[row sep=crcr]{%
10	0.000246392368631654\\
30	0.000965853851849298\\
50	0.00233268591538815\\
70	0.00477411654891502\\
90	0.008585507350914\\
110	0.0142075226021327\\
130	0.0219161833578853\\
};
\addlegendentry{LDG};

\addplot [color=mycolor1,dashed,line width=0.8pt,mark size=2.5pt,mark=x,mark options={solid}]
  table[row sep=crcr]{%
10	0.017067617033\\
30	0.125937311559\\
50	0.649451381534999\\
70	1.615938128329\\
90	3.895394904101\\
110	6.70871208363199\\
130	11.548104417797\\
};
\addlegendentry{DirAP};

\addplot [color=mycolor2,solid,line width=0.6pt,mark size=2.3pt,mark=o,mark options={solid}]
  table[row sep=crcr]{%
10	0.0443161864575908\\
30	0.530772029935852\\
50	3.03131674354954\\
70	8.0478871147541\\
90	20.5869972399144\\
110	38.6683249000713\\
130	82.7530849753386\\
};
\addlegendentry{APIndexCoding};

\end{axis}
\end{tikzpicture}%

%% file: ErrorFactor_p=0.2_iter=1000.tikz
% This file was created by matlab2tikz v0.5.0 running on MATLAB 8.1.
% Copyright (c) 2008--2014, Nico Schlömer <nico.schloemer@gmail.com>
% All rights reserved.
% Minimal pgfplots version: 1.3
% 
% The latest updates can be retrieved from
%   http://www.mathworks.com/matlabcentral/fileexchange/22022-matlab2tikz
% where you can also make suggestions and rate matlab2tikz.
% 
\begin{tikzpicture}

\begin{axis}[%
width=0.951322751322751\figurewidth,
height=\figureheight,
at={(0\figurewidth,0\figureheight)},
scale only axis,
xmin=10,
xmax=80,
xlabel={$n$},
ymode=log,
ymin=1e-05,
ymax=0.1,
yminorticks=true,
ylabel={Decoding Error},
legend style={at={(0.97,0.5)},anchor=east,draw=black,fill=white,legend cell align=left}
]
\addplot [color=blue,solid,line width=0.8pt,mark=triangle,mark options={solid}]
  table[row sep=crcr]{%
10	0.000205547265658617\\
20	0.000119904274025572\\
30	0.000101480379310953\\
40	9.6348707719209e-05\\
50	8.99581288678739e-05\\
60	8.4774115316734e-05\\
70	7.98628529292968e-05\\
80	7.16645482849152e-05\\
};
\addlegendentry{APIndexCoding Decoding Error};

\addplot [color=red,dash pattern=on 1pt off 3pt on 3pt off 3pt,line width=0.8pt,mark size=1.8pt,mark=square,mark options={solid}]
  table[row sep=crcr]{%
10	0.0316227766016838\\
20	0.0447213595499958\\
30	0.0547722557505166\\
40	0.0632455532033676\\
50	0.0707106781186548\\
60	0.0774596669241483\\
70	0.0836660026534076\\
80	0.0894427190999916\\
};
\addlegendentry{Theoretic Bound in Lemma 3};

\end{axis}
\end{tikzpicture}%

%% file: Directed_DirSVDAP_iter=100_Rmin.tikz
% This file was created by matlab2tikz v0.5.0 running on MATLAB 8.1.
% Copyright (c) 2008--2014, Nico Schlömer <nico.schloemer@gmail.com>
% All rights reserved.
% Minimal pgfplots version: 1.3
% 
% The latest updates can be retrieved from
%   http://www.mathworks.com/matlabcentral/fileexchange/22022-matlab2tikz
% where you can also make suggestions and rate matlab2tikz.
% 
%
% defining custom colors
\definecolor{mycolor1}{rgb}{0.00000,0.49804,0.00000}%
\begin{tikzpicture}

\begin{axis}[%
width=\figurewidth,
height=\figureheight,
at={(0\figurewidth,0\figureheight)},
scale only axis,
xmin=0,
xmax=1,
ymin=0,
ymax=1,
hide axis,
axis x line*=bottom,
axis y line*=left
]
\node[below, right, inner sep=0mm, text=red]
at (rel axis cs:0.628580563438227,0.318231390882494)(rel axis cs:0.0511234036798173,0.0438548812664901) {$p=0.6$};
\node[below, right, inner sep=0mm, text=mycolor1]
at (rel axis cs:0.608670621444617,0.756208017452157)(rel axis cs:0.0511234036798173,0.0438548812664901) {$p=0.2$};
\node[below, right, inner sep=0mm, text=blue]
at (rel axis cs:0.618353655064814,0.574514334251058)(rel axis cs:0.0511234036798173,0.0438548812664901) {$p=0.4$};
\end{axis}

\begin{axis}[%
width=0.775\figurewidth,
height=0.815\figureheight,
at={(0.13\figurewidth,0.113384094754653\figureheight)},
scale only axis,
xmin=10,
xmax=80,
xlabel={$n$},
ymin=0,
ymax=50,
ylabel={Average Index Code Length},
legend style={at={(0.03,0.97)},anchor=north west,draw=black,fill=white,legend cell align=left}
]
\addplot [color=mycolor1,dotted,line width=0.6pt,mark size=2.5pt,mark=x,mark options={solid},forget plot]
  table[row sep=crcr]{%
10	8.69\\
20	15.65\\
30	21.51\\
40	26.9\\
50	32.2\\
60	37.45\\
70	42.5409836065574\\
80	46.8571428571429\\
};
\addplot [color=mycolor1,solid,line width=0.6pt,mark size=2.0pt,mark=o,mark options={solid},forget plot]
  table[row sep=crcr]{%
10	8.27\\
20	14.36\\
30	20.53\\
40	26.25\\
50	32.07\\
60	37.41\\
70	42.9508196721311\\
80	48.4821428571429\\
};
\addplot [color=mycolor1,dashed,line width=0.6pt,mark size=1.7pt,mark=triangle,mark options={solid,rotate=180},forget plot]
  table[row sep=crcr]{%
10	8.53\\
20	15.51\\
30	21.15\\
40	25.86\\
50	30.5\\
60	34.55\\
70	38.7049180327869\\
80	42.0357142857143\\
};
\addplot [color=blue,dotted,line width=0.6pt,mark size=2.5pt,mark=x,mark options={solid},forget plot]
  table[row sep=crcr]{%
10	6.93\\
20	11.24\\
30	15.72\\
40	19.72\\
50	23.29\\
60	26.7\\
70	30.27\\
80	33.390243902439\\
};
\addplot [color=blue,solid,line width=0.6pt,mark size=2.0pt,mark=o,mark options={solid},forget plot]
  table[row sep=crcr]{%
10	6.21\\
20	10.17\\
30	14.46\\
40	18.44\\
50	22.1\\
60	25.63\\
70	29.24\\
80	33.0243902439024\\
};
\addplot [color=blue,dashed,line width=0.6pt,mark size=1.7pt,mark=triangle,mark options={solid,rotate=180},forget plot]
  table[row sep=crcr]{%
10	6.73\\
20	10.29\\
30	13.77\\
40	16.86\\
50	19.95\\
60	22.64\\
70	25.49\\
80	28.219512195122\\
};
\addplot [color=red,dotted,line width=0.6pt,mark size=2.5pt,mark=x,mark options={solid},forget plot]
  table[row sep=crcr]{%
10	5.24\\
20	8.54\\
30	11.23\\
40	13.91\\
50	16.21\\
60	18.74\\
70	20.99\\
80	23.26\\
};
\addplot [color=red,solid,line width=0.6pt,mark size=2.0pt,mark=o,mark options={solid},forget plot]
  table[row sep=crcr]{%
10	4.53\\
20	7.18\\
30	9.7\\
40	12.26\\
50	14.44\\
60	16.65\\
70	18.87\\
80	21.14\\
};
\addplot [color=red,dashed,line width=0.6pt,mark size=1.7pt,mark=triangle,mark options={solid,rotate=180},forget plot]
  table[row sep=crcr]{%
10	4.84\\
20	7.41\\
30	9.53\\
40	11.52\\
50	13.63\\
60	15.57\\
70	17.36\\
80	19.16\\
};
\addplot [color=black,dotted,line width=0.6pt,mark size=2.5pt,mark=x,mark options={solid}]
  table[row sep=crcr]{%
30	40\\
};
\addlegendentry{LDG on $G_d$};

\addplot [color=black,solid,line width=0.6pt,mark size=2.0pt,mark=o,mark options={solid}]
  table[row sep=crcr]{%
30	40\\
};
\addlegendentry{DirSVDAP  on $G_d$};

\addplot [color=black,dashed,line width=0.6pt,mark size=1.7pt,mark=triangle,mark options={solid,rotate=180}]
  table[row sep=crcr]{%
30	40\\
};
\addlegendentry{APIndexCoding on $G$};

\end{axis}
\end{tikzpicture}%

%% file: undirected_SVD_p=0.2_iter=500_Rmin.tikz
% This file was created by matlab2tikz v0.5.0 running on MATLAB 8.1.
% Copyright (c) 2008--2014, Nico Schlömer <nico.schloemer@gmail.com>
% All rights reserved.
% Minimal pgfplots version: 1.3
% 
% The latest updates can be retrieved from
%   http://www.mathworks.com/matlabcentral/fileexchange/22022-matlab2tikz
% where you can also make suggestions and rate matlab2tikz.
% 
%
% defining custom colors
\definecolor{mycolor1}{rgb}{0.00000,0.49804,0.00000}%
\begin{tikzpicture}

\begin{axis}[%
width=0.950920245398773\figurewidth,
height=\figureheight,
at={(0\figurewidth,0\figureheight)},
scale only axis,
xmin=10,
xmax=40,
xlabel={$n$},
ymin=6,
ymax=18,
ylabel={Average Index Code Length},
legend style={at={(0.03,0.97)},anchor=north west,draw=black,fill=white,legend cell align=left}
]
\addplot [color=blue,dotted,line width=0.8pt,mark size=2.8pt,mark=x,mark options={solid}]
  table[row sep=crcr]{%
10	6.092\\
20	9.848\\
30	13.346\\
40	16.824\\
};
\addlegendentry{Greedy Col.};

\addplot [color=mycolor1,solid,line width=0.8pt,mark size=2.8pt,mark=o,mark options={solid}]
  table[row sep=crcr]{%
10	6.182\\
20	10.128\\
30	13.374\\
40	16.456\\
};
\addlegendentry{SVDAP};

\addplot [color=red,dashed,line width=0.8pt,mark size=1.9pt,mark=square,mark options={solid}]
  table[row sep=crcr]{%
10	6.066\\
20	9.55\\
30	12.528\\
40	15.416\\
};
\addlegendentry{APIndexCoding};

\end{axis}
\end{tikzpicture}%

%% file: undirected_SVD_p=0.2_iter=500_Time.tikz
% This file was created by matlab2tikz v0.5.0 running on MATLAB 8.1.
% Copyright (c) 2008--2014, Nico Schlömer <nico.schloemer@gmail.com>
% All rights reserved.
% Minimal pgfplots version: 1.3
% 
% The latest updates can be retrieved from
%   http://www.mathworks.com/matlabcentral/fileexchange/22022-matlab2tikz
% where you can also make suggestions and rate matlab2tikz.
% 
%
% defining custom colors
\definecolor{mycolor1}{rgb}{0.00000,0.49804,0.00000}%
\begin{tikzpicture}

\begin{axis}[%
width=0.950920245398773\figurewidth,
height=\figureheight,
at={(0\figurewidth,0\figureheight)},
scale only axis,
xmin=10,
xmax=40,
xlabel={$n$},
ymin=0,
ymax=800,
ylabel={Average Running Time per Graph (sec.)},
legend style={at={(0.03,0.97)},anchor=north west,draw=black,fill=white,legend cell align=left}
]
\addplot [color=mycolor1,solid,line width=0.8pt,mark size=2.8pt,mark=o,mark options={solid}]
  table[row sep=crcr]{%
10	3.32598239555379\\
20	115.998212605876\\
30	402.831817777388\\
40	745.053935177617\\
};
\addlegendentry{SVDAP};

\addplot [color=red,dashed,line width=0.8pt,mark size=1.9pt,mark=square,mark options={solid}]
  table[row sep=crcr]{%
10	0.0172033523372477\\
20	0.0953963620480892\\
30	0.504070875657836\\
40	1.91836036770108\\
};
\addlegendentry{APIndexCoding};

\end{axis}
\end{tikzpicture}%

%% file: undirected_AM_p=0.2_iter=100_Rmin.tikz
% This file was created by matlab2tikz v0.5.0 running on MATLAB 8.1.
% Copyright (c) 2008--2014, Nico Schlömer <nico.schloemer@gmail.com>
% All rights reserved.
% Minimal pgfplots version: 1.3
% 
% The latest updates can be retrieved from
%   http://www.mathworks.com/matlabcentral/fileexchange/22022-matlab2tikz
% where you can also make suggestions and rate matlab2tikz.
% 
%
% defining custom colors
\definecolor{mycolor1}{rgb}{0.00000,0.49804,0.00000}%
\begin{tikzpicture}

\begin{axis}[%
width=0.950920245398773\figurewidth,
height=\figureheight,
at={(0\figurewidth,0\figureheight)},
scale only axis,
xmin=10,
xmax=70,
xlabel={$n$},
ymin=5,
ymax=30,
ylabel={Average Index Code Length},
legend style={at={(0.03,0.97)},anchor=north west,draw=black,fill=white,legend cell align=left}
]
\addplot [color=mycolor1,dotted,line width=0.8pt,mark size=1.7pt,mark=square,mark options={solid}]
  table[row sep=crcr]{%
10	6.061\\
20	9.912\\
30	13.368\\
40	16.753\\
50	20.022\\
60	23.066\\
70	26.113\\
};
\addlegendentry{Greedy Col.};

\addplot [color=red,dashed,line width=0.8pt,mark=x,mark options={solid}]
  table[row sep=crcr]{%
10	6.28\\
20	9.72\\
30	12.92\\
40	15.92\\
50	18.6\\
60	21.48\\
70	24.3333333333333\\
};
\addlegendentry{AltMin};

\addplot [color=blue,solid,line width=0.8pt,mark size=2.2pt,mark=o,mark options={solid}]
  table[row sep=crcr]{%
10	6.033\\
20	9.567\\
30	12.551\\
40	15.425\\
50	18.193\\
60	20.849\\
70	23.38\\
};
\addlegendentry{DirAP};

\end{axis}
\end{tikzpicture}%

%% file: AltPrj_vs_noncoding.tikz
% This file was created by matlab2tikz v0.5.0 running on MATLAB 8.1.
% Copyright (c) 2008--2014, Nico Schlömer <nico.schloemer@gmail.com>
% All rights reserved.
% Minimal pgfplots version: 1.3
% 
% The latest updates can be retrieved from
%   http://www.mathworks.com/matlabcentral/fileexchange/22022-matlab2tikz
% where you can also make suggestions and rate matlab2tikz.
% 
%
% defining custom colors
\definecolor{mycolor1}{rgb}{0.00000,0.49804,0.00000}%
\definecolor{mycolor2}{rgb}{0.60000,0.20000,0.00000}%
\begin{tikzpicture}

\begin{axis}[%
width=\figurewidth,
height=0.981711197994467\figureheight,
at={(0\figurewidth,0\figureheight)},
scale only axis,
xmin=10,
xmax=100,
xlabel={$n$},
ymin=30,
ymax=100,
ytick={ 40,  60,  80, 100},
ylabel={APIndexCoding Savings over uncoded in \%},
legend style={at={(0.97,0.03)},anchor=south east,draw=black,fill=white,legend cell align=left}
]
\addplot [color=blue,dotted,line width=0.8pt,mark size=1.8pt,mark=square,mark options={solid}]
  table[row sep=crcr]{%
10	39.67\\
20	52.165\\
30	58.1633333333333\\
40	61.4375\\
50	63.614\\
60	65.2516666666667\\
70	66.6\\
80	67.70375\\
90	68.51\\
100	69.242\\
};
\addlegendentry{p=0.2};

\addplot [color=red,dash pattern=on 1pt off 3pt on 3pt off 3pt,line width=0.8pt,mark=triangle,mark options={solid}]
  table[row sep=crcr]{%
10	54.55\\
20	65.74\\
30	70.33\\
40	72.92\\
50	74.66\\
60	75.9383333333333\\
70	76.8742857142857\\
80	77.625\\
90	78.2477777777778\\
100	78.756\\
};
\addlegendentry{p=0.4};

\addplot [color=mycolor1,dashed,line width=0.8pt,mark=x,mark options={solid}]
  table[row sep=crcr]{%
10	64.7888888888889\\
20	74.8111111111111\\
30	78.7925925925926\\
40	80.9\\
50	82.2733333333333\\
60	83.2722222222222\\
70	84.1047619047619\\
80	84.675\\
90	85.2024691358025\\
100	85.5911111111111\\
};
\addlegendentry{p=0.6};

\addplot [color=mycolor2,solid,line width=0.8pt,mark size=2.5pt,mark=o,mark options={solid}]
  table[row sep=crcr]{%
10	73.79\\
20	82.905\\
30	86.18\\
40	87.72\\
50	88.974\\
60	89.795\\
70	90.3485714285714\\
80	90.85125\\
90	91.1906594816526\\
100	91.54\\
};
\addlegendentry{p=0.8};

\end{axis}
\end{tikzpicture}%

%% file: AltPrj_vs_GreedyColor_iter=1000.tikz
% This file was created by matlab2tikz v0.5.0 running on MATLAB 8.1.
% Copyright (c) 2008--2014, Nico Schlömer <nico.schloemer@gmail.com>
% All rights reserved.
% Minimal pgfplots version: 1.3
% 
% The latest updates can be retrieved from
%   http://www.mathworks.com/matlabcentral/fileexchange/22022-matlab2tikz
% where you can also make suggestions and rate matlab2tikz.
% 
%
% defining custom colors
\definecolor{mycolor1}{rgb}{0.00000,0.49804,0.00000}%
\definecolor{mycolor2}{rgb}{0.60000,0.20000,0.00000}%
\begin{tikzpicture}

\begin{axis}[%
width=\figurewidth,
height=0.975877165314594\figureheight,
at={(0\figurewidth,0\figureheight)},
scale only axis,
xmin=10,
xmax=100,
xlabel={$n$},
ymin=0,
ymax=15,
ylabel={APIndexCoding Savings over Greedy Col. in \%},
legend style={at={(0.97,0.03)},anchor=south east,draw=black,fill=white,legend cell align=left}
]
\addplot [color=blue,dotted,line width=0.8pt,mark size=1.8pt,mark=square,mark options={solid}]
  table[row sep=crcr]{%
10	0.461969971951816\\
20	3.48062953995158\\
30	6.11160981448235\\
40	7.9269384587835\\
50	9.13495155329137\\
60	9.61154946674759\\
70	10.466051392027\\
80	10.9744331886155\\
90	11.0954263128176\\
100	11.2271992611406\\
};
\addlegendentry{p=0.2};

\addplot [color=red,dash pattern=on 1pt off 3pt on 3pt off 3pt,line width=0.8pt,mark=triangle,mark options={solid}]
  table[row sep=crcr]{%
10	2.06851971557854\\
20	6.2012320328542\\
30	8.55763303883296\\
40	9.72581048420701\\
50	10.472018089316\\
60	10.8001235712079\\
70	11.1086705837131\\
80	10.8743278231428\\
90	11.0661881615409\\
100	11.0012568077084\\
};
\addlegendentry{p=0.4};

\addplot [color=mycolor1,dashed,line width=0.8pt,mark=x,mark options={solid}]
  table[row sep=crcr]{%
10	3.02937576499388\\
20	7.18526100307063\\
30	9.81256890848952\\
40	10.6664934390022\\
50	11.2779446112779\\
60	11.3890523837551\\
70	11.6541685046317\\
80	11.5440115440115\\
90	11.5032486709982\\
100	11.4872704934817\\
};
\addlegendentry{p=0.6};

\addplot [color=mycolor2,solid,line width=0.8pt,mark size=2.5pt,mark=o,mark options={solid}]
  table[row sep=crcr]{%
10	1.50319428786171\\
20	9.38245428041347\\
30	11.9558292631132\\
40	11.5751575157516\\
50	12.6307448494453\\
60	13.3087922978904\\
70	12.8145567169957\\
80	13.5993389210247\\
90	12.7794715447154\\
100	13.0226182316655\\
};
\addlegendentry{p=0.8};

\end{axis}
\end{tikzpicture}%

%% file: AltPrj_vs_LDG_iter=1000.tikz
% This file was created by matlab2tikz v0.5.0 running on MATLAB 8.1.
% Copyright (c) 2008--2014, Nico Schlömer <nico.schloemer@gmail.com>
% All rights reserved.
% Minimal pgfplots version: 1.3
% 
% The latest updates can be retrieved from
%   http://www.mathworks.com/matlabcentral/fileexchange/22022-matlab2tikz
% where you can also make suggestions and rate matlab2tikz.
% 
%
% defining custom colors
\definecolor{mycolor1}{rgb}{0.00000,0.49804,0.00000}%
\definecolor{mycolor2}{rgb}{0.60000,0.20000,0.00000}%
\begin{tikzpicture}

\begin{axis}[%
width=\figurewidth,
height=0.975877165314594\figureheight,
at={(0\figurewidth,0\figureheight)},
scale only axis,
xmin=10,
xmax=100,
xlabel={$n$},
ymin=0,
ymax=15,
ylabel={APIndexCoding Savings over LDG in \%},
legend style={at={(0.97,0.03)},anchor=south east,draw=black,fill=white,legend cell align=left}
]
\addplot [color=blue,dotted,line width=0.8pt,mark size=1.8pt,mark=square,mark options={solid}]
  table[row sep=crcr]{%
10	1.16317169069462\\
20	4.59712804148385\\
30	5.89337932068681\\
40	6.44143870928609\\
50	6.38571575589173\\
60	6.5318748318838\\
70	6.94156981372393\\
80	6.93729063861975\\
90	6.968881302521\\
100	6.94339394306114\\
};
\addlegendentry{p=0.2};

\addplot [color=red,dash pattern=on 1pt off 3pt on 3pt off 3pt,line width=0.8pt,mark size=1.7pt,mark=triangle,mark options={solid}]
  table[row sep=crcr]{%
10	3.27729304107257\\
20	5.82737768004397\\
30	7.04887218045114\\
40	7.25233324770956\\
50	7.19308526223263\\
60	7.05594540655379\\
70	6.85845799769851\\
80	6.41501542322372\\
90	6.26286808714388\\
100	6.05819403909083\\
};
\addlegendentry{p=0.4};

\addplot [color=mycolor1,dashed,line width=0.8pt,mark size=2.5pt,mark=x,mark options={solid}]
  table[row sep=crcr]{%
10	3.44302254722729\\
20	7.95777507105157\\
30	9.01001112347052\\
40	8.67313056182761\\
50	8.35248161764706\\
60	8.16388775925172\\
70	8.11158010644154\\
80	7.417351904682\\
90	7.18599969025863\\
100	6.53020037480178\\
};
\addlegendentry{p=0.6};

\addplot [color=mycolor2,solid,line width=0.8pt,mark size=2.3pt,mark=o,mark options={solid}]
  table[row sep=crcr]{%
10	1.9453797231575\\
20	9.90777338603425\\
30	11.561433447099\\
40	10.6259097525473\\
50	11.9610348131587\\
60	11.873920552677\\
70	10.8589523683863\\
80	11.1016640349812\\
90	10.5523710265763\\
100	10.2228510788822\\
};
\addlegendentry{p=0.8};

\end{axis}
\end{tikzpicture}%

%% file: AltPrj_Time_of_Iteration.tikz
% This file was created by matlab2tikz v0.5.0 running on MATLAB 8.1.
% Copyright (c) 2008--2014, Nico Schlömer <nico.schloemer@gmail.com>
% All rights reserved.
% Minimal pgfplots version: 1.3
% 
% The latest updates can be retrieved from
%   http://www.mathworks.com/matlabcentral/fileexchange/22022-matlab2tikz
% where you can also make suggestions and rate matlab2tikz.
% 
%
% defining custom colors
\definecolor{mycolor1}{rgb}{0.00000,0.49804,0.00000}%
\definecolor{mycolor2}{rgb}{0.60000,0.20000,0.00000}%
\begin{tikzpicture}

\begin{axis}[%
width=0.950920245398773\figurewidth,
height=\figureheight,
at={(0\figurewidth,0\figureheight)},
scale only axis,
xmin=10,
xmax=100,
xlabel={$n$},
ymin=0,
ymax=0.0035,
ylabel={Average Running Time per Iteration},
legend style={at={(0.03,0.97)},anchor=north west,draw=black,fill=white,legend cell align=left}
]
\addplot [color=blue,dotted,line width=0.8pt,mark=square,mark options={solid}]
  table[row sep=crcr]{%
10	2.83063625541989e-05\\
20	8.65014122128179e-05\\
30	0.00019131001277961\\
40	0.000380417483028863\\
50	0.000585716756327976\\
60	0.000810925880286377\\
70	0.00121003026843474\\
80	0.0018661746426615\\
90	0.00234124597016261\\
100	0.00293754670978788\\
};
\addlegendentry{p=0.2};

\addplot [color=red,dash pattern=on 1pt off 3pt on 3pt off 3pt,line width=0.8pt,mark=triangle,mark options={solid}]
  table[row sep=crcr]{%
10	3.3330858156128e-05\\
20	8.81473072122001e-05\\
30	0.000202580171305845\\
40	0.000420149389772727\\
50	0.000671839318857035\\
60	0.000887518686411197\\
70	0.0013052355346291\\
80	0.00190381443877226\\
90	0.00243063101010951\\
100	0.00287145155706741\\
};
\addlegendentry{p=0.4};

\addplot [color=mycolor1,dashed,line width=0.8pt,mark=x,mark options={solid}]
  table[row sep=crcr]{%
10	3.23451364315386e-05\\
20	9.02867687808608e-05\\
30	0.000210263150908481\\
40	0.000424985622178667\\
50	0.000663028141503632\\
60	0.000949788345715222\\
70	0.00140651769647174\\
80	0.00203741999461494\\
90	0.00262793502002326\\
100	0.00321603849418877\\
};
\addlegendentry{p=0.6};

\addplot [color=mycolor2,solid,line width=0.8pt,mark=o,mark options={solid}]
  table[row sep=crcr]{%
10	3.0907783240657e-05\\
20	9.74104854208688e-05\\
30	0.000207556758622956\\
40	0.000421415086215678\\
50	0.000675267100375473\\
60	0.000962025058274185\\
70	0.0014785912064701\\
80	0.00192221370406396\\
90	0.00245501647941408\\
100	0.00338201404584767\\
};
\addlegendentry{p=0.8};

\end{axis}
\end{tikzpicture}%

%% file: AltPrj_ITE.tikz
% This file was created by matlab2tikz v0.5.0 running on MATLAB 8.1.
% Copyright (c) 2008--2014, Nico Schlömer <nico.schloemer@gmail.com>
% All rights reserved.
% Minimal pgfplots version: 1.3
% 
% The latest updates can be retrieved from
%   http://www.mathworks.com/matlabcentral/fileexchange/22022-matlab2tikz
% where you can also make suggestions and rate matlab2tikz.
% 
%
% defining custom colors
\definecolor{mycolor1}{rgb}{0.00000,0.49804,0.00000}%
\definecolor{mycolor2}{rgb}{0.60000,0.20000,0.00000}%
\begin{tikzpicture}

\begin{axis}[%
width=0.950920245398773\figurewidth,
height=\figureheight,
at={(0\figurewidth,0\figureheight)},
scale only axis,
xmin=10,
xmax=100,
xlabel={$n$},
ymin=0,
ymax=120000,
ylabel={Average Number of Iterations},
legend style={at={(0.03,0.97)},anchor=north west,draw=black,fill=white,legend cell align=left}
]
\addplot [color=blue,dotted,line width=0.8pt,mark=square,mark options={solid}]
  table[row sep=crcr]{%
10	1565.591\\
20	1371.574\\
30	2774.408\\
40	3856.924\\
50	5175.397\\
60	5485.08\\
70	6650.98\\
80	7707.76\\
90	8793.18\\
100	9771.7\\
};
\addlegendentry{p=0.2};

\addplot [color=red,dash pattern=on 1pt off 3pt on 3pt off 3pt,line width=0.8pt,mark=triangle,mark options={solid}]
  table[row sep=crcr]{%
10	1221.11111111111\\
20	3190.79555555556\\
30	5338.76666666667\\
40	7250.17777777778\\
50	9047.77777777778\\
60	11746.2\\
70	11007.2666666667\\
80	13490.0444444444\\
90	17081\\
100	16800.6666666667\\
};
\addlegendentry{p=0.4};

\addplot [color=mycolor1,dashed,line width=0.8pt,mark=x,mark options={solid}]
  table[row sep=crcr]{%
10	1555.29222222222\\
20	5544.33333333333\\
30	9940.96666666667\\
40	13971.2\\
50	16645.0555555556\\
60	22196.4666666667\\
70	25711\\
80	30545.2444444444\\
90	32646\\
100	31657.5555555556\\
};
\addlegendentry{p=0.6};

\addplot [color=mycolor2,solid,line width=0.8pt,mark=o,mark options={solid}]
  table[row sep=crcr]{%
10	2234.43\\
20	9414.22\\
30	29020.35\\
40	42950.96\\
50	63171.7\\
60	70827.66\\
70	86025.8\\
80	89377.84\\
90	99524.4110854504\\
100	100216.333333333\\
};
\addlegendentry{p=0.8};

\end{axis}
\end{tikzpicture}%

%% file: Directed_cache_AltPrj_Rmin.tikz
% This file was created by matlab2tikz v0.5.0 running on MATLAB 8.1.
% Copyright (c) 2008--2014, Nico Schlömer <nico.schloemer@gmail.com>
% All rights reserved.
% Minimal pgfplots version: 1.3
% 
% The latest updates can be retrieved from
%   http://www.mathworks.com/matlabcentral/fileexchange/22022-matlab2tikz
% where you can also make suggestions and rate matlab2tikz.
% 
\begin{tikzpicture}

\begin{axis}[%
width=\figurewidth,
height=\figureheight,
at={(0\figurewidth,0\figureheight)},
scale only axis,
xmin=0,
xmax=1,
ymin=0,
ymax=1,
hide axis,
axis x line*=bottom,
axis y line*=left
]
\node[below, right, inner sep=0mm, text=black, draw=white]
at (rel axis cs:0.751707203474581,0.858257481186416)(rel axis cs:0.0428751893019417,0.0474934036939302) {$c=8$};
\node[below, right, inner sep=0mm, text=black, draw=white]
at (rel axis cs:0.750262508214984,0.679435464725512)(rel axis cs:0.0375478755321907,0.0474934036939302) {$c=16$};
\node[below, right, inner sep=0mm, text=black, draw=white]
at (rel axis cs:0.781865217018595,0.399214271309872)(rel axis cs:0.0428751893019417,0.0474934036939302) {$c=32$};
\end{axis}

\begin{axis}[%
width=0.775\figurewidth,
height=0.815\figureheight,
at={(0.13\figurewidth,0.11\figureheight)},
scale only axis,
xmin=10,
xmax=100,
xlabel={$n$},
ymin=0,
ymax=80,
ylabel={Average Index Code length},
legend style={at={(0.03,0.97)},anchor=north west,draw=black,fill=white,legend cell align=left}
]
\addplot [color=blue,dotted,mark size=1.5pt,mark=triangle,mark options={solid}]
  table[row sep=crcr]{%
10	2.574\\
20	10.094\\
30	17.4\\
40	25.17\\
50	33.536\\
60	42.286\\
70	51.276\\
80	60.32\\
90	69.372\\
100	78.576\\
};
\addlegendentry{Greedy Col.};

\addplot [color=black!50!green,dashed,mark size=2.0pt,mark=x,mark options={solid}]
  table[row sep=crcr]{%
10	2.56\\
20	10.16\\
30	17.718\\
40	25.426\\
50	33.78\\
60	42.492\\
70	51.4\\
80	60.428\\
90	69.494\\
100	78.648\\
};
\addlegendentry{LDG};

\addplot [color=red,solid,mark size=1.9pt,mark=o,mark options={solid}]
  table[row sep=crcr]{%
10	2.464\\
20	9.704\\
30	17.046\\
40	24.944\\
50	33.372\\
60	42.214\\
70	51.21\\
80	60.274\\
90	69.358\\
100	78.56\\
};
\addlegendentry{APIndexCoding};

\addplot [color=blue,dotted,mark size=1.5pt,mark=triangle,mark options={solid},forget plot]
  table[row sep=crcr]{%
20	4.438\\
30	11.274\\
40	17.772\\
50	24.138\\
60	30.344\\
70	36.656\\
80	43.164\\
90	50.006\\
100	56.988\\
};
\addplot [color=black!50!green,dashed,mark size=2.0pt,mark=x,mark options={solid},forget plot]
  table[row sep=crcr]{%
20	4.414\\
30	10.988\\
40	17.354\\
50	23.936\\
60	30.708\\
70	37.372\\
80	44.212\\
90	51.202\\
100	58.292\\
};
\addplot [color=red,solid,mark size=1.9pt,mark=o,mark options={solid},forget plot]
  table[row sep=crcr]{%
20	4.044\\
30	10.148\\
40	16.194\\
50	22.314\\
60	28.538\\
70	34.718\\
80	41.234\\
90	48.006\\
100	55.118\\
};
\addplot [color=blue,dotted,mark size=1.5pt,mark=triangle,mark options={solid},forget plot]
  table[row sep=crcr]{%
40	7.436\\
50	13.642\\
60	19.438\\
70	25.408\\
80	31.398\\
90	37.52\\
100	43.688\\
};
\addplot [color=black!50!green,dashed,mark size=2.0pt,mark=x,mark options={solid},forget plot]
  table[row sep=crcr]{%
40	7.178\\
50	12.974\\
60	18.602\\
70	24.202\\
80	29.994\\
90	35.682\\
100	41.584\\
};
\addplot [color=red,solid,mark size=1.9pt,mark=o,mark options={solid},forget plot]
  table[row sep=crcr]{%
40	6.396\\
50	11.86\\
60	17.102\\
70	22.336\\
80	27.614\\
90	32.938\\
100	38.366\\
};
\end{axis}
\end{tikzpicture}%

%% file: ICviaConvexOptimization_Arxiv.bbl
\begin{thebibliography}{10}

\bibitem{BirkKol}
Y.~Birk and T.~Kol, ``Informed-source coding-on-demand ({ISCOD}) over broadcast
  channels,'' {\em INFOCOM}, vol.~3, pp.~1257--1264, 1998.

\bibitem{ISCOD2006}
Y.~Birk and T.~Kol., ``Coding on demand by an informed source ({ISCOD}) for
  efficient broadcast of different supplemental data to caching clients,'' {\em
  IEEE Transactions on Information Theory}, vol.~52, pp.~2825--2830, June 2006.

\bibitem{RSG10}
S.~{E}l {R}ouayheb, A.~Sprintson, and C.~Georghiades, ``On the index coding
  problem and its relation to network coding and matroid theory,'' {\em {IEEE}
  Transactions on Information Theory}, vol.~56, July 2010.

\bibitem{EEL13}
M.~Effros, S.~{E}l {R}ouayheb, and M.~Langberg, ``{An Equivalence between
  Network Coding and Index Coding},'' {\em {IEEE} International Symposium on
  Information Theory}, pp.~967--971, 2013.

\bibitem{BBJK06}
Z.~Bar-Yossef, Y.~Birk, T.~S. Jayram, and T.~Kol, ``Index {C}oding with {S}ide
  {I}nformation,'' {\em In Proceedings of 47th Annual IEEE Symposium on
  Foundations of Computer Science}, pp.~197--206, 2006.

\bibitem{RCA07}
S.~{E}l {R}ouayheb, M.~A.~R. Chaudhry, and A.~Sprintson, ``On the minimum
  number of transmissions in single-hop wireless coding networks,'' in {\em
  Information Theory Workshop (ITW)}, 2007.

\bibitem{Pee96}
R.~Peeters, ``Orthogonal {R}epresentations {O}ver {F}inite {F}ields and the
  {C}hromatic {N}umber of {G}raphs,'' {\em Combinatorica}, vol.~16, no.~3,
  pp.~417--431, 1996.

\bibitem{LangbergSprintson11}
M.~A.~R. Chaudhry, Z.~Asad, A.~Sprintson, and M.~Langberg, ``On the
  complementary index coding problem,'' in {\em IEEE International Symposium on
  Information Theory}, pp.~244--248, IEEE, 2011.

\bibitem{candes2009exact}
E.~J. Cand{\`e}s and B.~Recht, ``Exact matrix completion via convex
  optimization,'' {\em Foundations of Computational mathematics}, vol.~9,
  no.~6, pp.~717--772, 2009.

\bibitem{recht2010guaranteed}
B.~Recht, M.~Fazel, and P.~A. Parrilo, ``Guaranteed minimum-rank solutions of
  linear matrix equations via nuclear norm minimization,'' {\em SIAM review},
  vol.~52, no.~3, pp.~471--501, 2010.

\bibitem{Syed}
H.~Maleki, V.~R. Cadambe, and S.~A. Jafar, ``Index coding: An interference
  alignment perspective,'' in {\em International Symposium on Information
  Theory}, 2012.

\bibitem{jafar2013topological}
S.~A. Jafar, ``Topological interference management through index coding,'' {\em
  Information Theory, IEEE Transactions on}, vol.~60, no.~1, pp.~529--568,
  2013.

\bibitem{Jaggi08}
B.~K. Dey, S.~Katti, S.~Jaggi, D.~Katabi, M.~M{\'e}dard, and S.~Shintre,
  ````{R}eal'' and ``{C}omplex'' {N}etwork {C}odes: Promises and
  {C}hallenges,'' {\em Network Coding, Theory and Applications. NetCod 2008.
  Fourth Workshop on}, pp.~1--6, January 2008.

\bibitem{Gastpar12}
N.~Goela and M.~Gastpar, ``Reduced-dimension linear transform coding of
  correlated signals in networks,'' {\em IEEE TRANSACTIONS ON SIGNAL
  PROCESSING}, vol.~60, no.~6, 2012.

\bibitem{Arya}
A.~Mazumdar, ``On a duality between recoverable distributed storage and index
  coding,'' in {\em IEEE International Symposium on Information Theory (ISIT)},
  pp.~1977--1981, July 2003.

\bibitem{Dimakis14}
K.~Shanmugam and A.~G. Dimakis, ``Bounding multiple unicasts through index
  coding and locally repairable codes,'' in {\em IEEE International Symposium
  on Information Theory}, pp.~296--231, IEEE, 2013.

\bibitem{APlink}
X.~Huang and S.~{E}l {R}ouayheb, ``{API}ndex{C}oding {M}atlab {C}ode.''
  \url{http://www.ece.iit.edu/~salim/software.html}, 2015.

\bibitem{LS07}
E.~Lubetzky and U.~Stav, ``Non-linear {I}ndex {C}oding {O}utperforming the
  {L}inear {O}ptimum,'' {\em In Proceedings of 48th Annual IEEE Symposium on
  Foundations of Computer Science}, pp.~161--168, 2007.

\bibitem{RSG08}
S.~{E}l {R}ouayheb, A.~Sprintson, and C.~Georghiades, ``{On the Relation
  Between the Index Coding and the Network Coding Problems},'' {\em In
  proceedings of IEEE International Symposium on Information Theory (ISIT)},
  2008.

\bibitem{BKL11}
A.~Blasiak, R.~Kleinberg, and E.~Lubetzky, ``Lexicographic products and the
  power of non-linear network coding,'' {\em In Proceedings of 52nd Annual IEEE
  Symposium on Foundations of Computer Science}, pp.~609--618, 2011.

\bibitem{ALSW08}
N.~Alon, E.~Lubetzky, U.~Stav, A.~Weinstein, and A.~Hassidim, ``Broadcasting
  with side information.,'' {\em In Proceedings of 49th Annual IEEE Symposium
  on Foundations of Computer Science}, pp.~823--832, 2008.

\bibitem{blasiak2010index}
A.~Blasiak, R.~Kleinberg, and E.~Lubetzky, ``Index coding via linear
  programming,'' in {\em arXiv preprint arXiv:1004.1379}, 2010.

\bibitem{DimakisLocal13}
K.~Shanmugam, A.~G. Dimakis, and M.~Langberg, ``Local graph coloring and index
  coding,'' in {\em International Symposium on Information Theory}, 2013.

\bibitem{haviv2012linear}
I.~Haviv and M.~Langberg, ``On linear index coding for random graphs,'' in {\em
  IEEE International Symposium on Information Theory}, pp.~2231--2235, IEEE,
  2012.

\bibitem{KimIndex13}
F.~Arbabjolfaei, B.~Bandemer, Y.-H. Kim, E.~Sasoglu, and L.~Wang, ``On the
  capacity region for index coding,'' in {\em IEEE International Symposium on
  Information Theory}, pp.~962--966, IEEE, 2013.

\bibitem{Kim14}
F.~Arbabjolfaei and Y.-H.~a. Kim, ``Local time sharing for index coding,'' in
  {\em 2014 IEEE International Symposium on Information Theory}, pp.~286--290,
  IEEE, 2014.

\bibitem{unal2013general}
S.~Unal and A.~B. Wagner, ``General index coding with side information: Three
  decoder case,'' in {\em IEEE International Symposium on Information Theory},
  pp.~1137--1141, IEEE, 2013.

\bibitem{MohammadCaching}
M.~Maddah-Ali and U.~Niesen, ``Fundamental limits of caching,'' in {\em
  International Symposium on Information Theory}, 2013.

\bibitem{maddah2013decentralized}
M.~A. Maddah-Ali and U.~Niesen, ``Decentralized coded caching attains
  order-optimal memory-rate tradeoff,'' {\em arXiv preprint arXiv:1301.5848},
  2013.

\bibitem{Hassibi2015}
H.~Esfahanizadeh, F.~Lahouti, and B.~Hassibi, ``A matrix completion approach to
  linear index coding problem,'' {\em Information Theory Workshop (ITW), IEEE},
  pp.~531--535, 2014.

\bibitem{Schwartz14}
M.~Schwartz and M.~M\'{e}dard, ``Quasi-linear network coding,'' {\em
  International Symposium on Network Coding}, 2014.

\bibitem{jafar2012elements}
S.~A. Jafar, ``Elements of cellular blind interference alignment---aligned
  frequency reuse, wireless index coding and interference diversity,'' {\em
  arXiv preprint arXiv:1203.2384}, 2012.

\bibitem{Karp}
R.~M. Karp, ``Reducibility among combinatorial problems,'' {\em Proc. Symp.
  Complexity of Computer Computations}, pp.~85--103, 1972.

\bibitem{garey2002computers}
M.~R. Garey and D.~S. Johnson, {\em Computers and intractability}.
\newblock Macmillan Higher Education, 1978.

\bibitem{shannon1956zero}
C.~E. Shannon, ``The zero error capacity of a noisy channel,'' {\em Information
  Theory, IRE Transactions on}, vol.~2, no.~3, pp.~8--19, 1956.

\bibitem{sun2013index}
H.~Sun and S.~A. Jafar, ``Index coding capacity: How far can one go with only
  shannon inequalities?,'' {\em arXiv preprint arXiv:1303.7000}, 2013.

\bibitem{fazel2001rank}
M.~Fazel, H.~Hindi, and S.~P. Boyd, ``A rank minimization heuristic with
  application to minimum order system approximation,'' in {\em Proceedings of
  the 2001 American Control Conference}, vol.~6, pp.~4734--4739, IEEE, 2001.

\bibitem{fazel2004rank}
M.~Fazel, H.~Hindi, and S.~Boyd, ``Rank minimization and applications in system
  theory,'' in {\em American Control Conference}, vol.~4, pp.~3273--3278, IEEE,
  2004.

\bibitem{AlternatingBoyd2003}
S.~Boyd and J.~Dattorro, ``Alternating projections.'' EE392o, Stanford
  University, 2003.

\bibitem{Bre67}
L.~M. Bregman, ``The relaxation method of finding the common point of convex
  sets and its application to the solution of problems in convex programming,''
  {\em USSR Computational Math. and Math. Physics}, vol.~7, no.~3,
  pp.~200--217, 1967.

\bibitem{Eckart36}
G.~Y. C.~Eckart, ``The approximation of one matrix by another of lower rank,''
  {\em Psychometrika}, vol.~1, pp.~211--218, September 1936.

\bibitem{jain2013low}
P.~Jain, P.~Netrapalli, and S.~Sanghavi, ``Low-rank matrix completion using
  alternating minimization,'' in {\em Proceedings of the forty-fifth annual ACM
  symposium on Theory of computing}, pp.~665--674, ACM, 2013.

\bibitem{Hardt2014}
M.~Hardt, ``Understanding alternating minimization for matrix completion,''
  {\em Foundations of Computer Science (FOCS), IEEE 55th Annual Symposium on},
  pp.~651--660, 2014.

\bibitem{AloSpe92}
N.~Alon and J.~H. Spencer, {\em The Probabilistic Method}.
\newblock Wiley-Interscience publication, 1992.

\bibitem{Kimurl}
\url{http://circuit.ucsd.edu/~yhk/indexcoding.html}.

\bibitem{Control1987}
L.~E. Ghaoui and S.~lulian Niculescu, {\em Advances in Linear Matrix Inequality
  Methods in Control}.
\newblock the Society for Industrial and Applied Mathematics., 1987.

\bibitem{EminaMonograph}
C.~Fragouli and E.~Soljanin, ``Monograph on network coding: Fundamentals and
  applications,'' {\em Foundations and Trends in Networking}, vol.~2, no.~1,
  2007.

\bibitem{ho2008network}
T.~Ho and D.~Lun, {\em Network coding: an introduction}.
\newblock Cambridge University Press, 2008.

\bibitem{Y10}
R.~W. Yeung, ``{Information Theory and Network Coding},'' {\em Springer}, 2008.

\bibitem{JagSanCEEJT}
S.~Jaggi, P.~Sanders, P.~A. Chou, M.~Effros, S.~Egner, K.~Jain, and
  L.~Tolhuizen, ``Polynomial time algorithms for multicast network code
  construction,'' {\em IEEE Transactions on Information Theory}, vol.~51,
  pp.~1973--1982, June 2005.

\bibitem{KM03}
R.~Koetter and M.~M\'{e}dard, ``An {A}lgebraic {A}pproach to {N}etwork
  {C}oding,'' {\em IEEE/ACM Transactions on Networking}, vol.~11, no.~5,
  pp.~782--795, 2003.

\bibitem{HoMKKESL06}
T.~Ho, M.~M{\'e}dard, R.~Koetter, D.~R. Karger, M.~Effros, J.~Shi, and
  B.~Leong, ``A {R}andom {L}inear {N}etwork {C}oding {A}pproach to
  {M}ulticast,'' {\em IEEE Transactions on Information Theory}, vol.~52,
  no.~10, pp.~4413--4430, 2006.

\bibitem{lun2006minimum}
D.~S. Lun, N.~Ratnakar, M.~M{\'e}dard, R.~Koetter, D.~R. Karger, T.~Ho,
  E.~Ahmed, and F.~Zhao, ``Minimum-cost multicast over coded packet networks,''
  {\em Information Theory, IEEE Transactions on}, vol.~52, no.~6,
  pp.~2608--2623, 2006.

\bibitem{DFZ05}
R.~Dougherty, C.~Freiling, and K.~Zeger, ``Insufficiency of {L}inear {C}oding
  in {N}etwork {I}nformation {F}low,'' {\em IEEE Transactions on Information
  Theory}, vol.~51, no.~8, pp.~2745--2759, 2005.

\bibitem{DFZ07}
R.~Dougherty, C.~Freiling, and K.~Zeger, ``{Networks, Matroids, and Non-Shannon
  Information Inequalities},'' {\em IEEE Transactions on Information Theory},
  vol.~53, no.~6, pp.~1949--1969, 2007.

\bibitem{Lehman2004complexity}
A.~R. Lehman and E.~Lehman, ``Complexity classification of network information
  flow problems,'' in {\em Proceedings of the fifteenth annual ACM-SIAM
  symposium on Discrete algorithms}, pp.~142--150, Society for Industrial and
  Applied Mathematics, 2004.

\bibitem{medard2003coding}
M.~M{\'e}dard, M.~Effros, D.~Karger, and T.~Ho, ``On coding for non-multicast
  networks,'' in {\em Proceedings of the Annual Allerton Conference on
  Communication Control and Computing}, vol.~41, pp.~21--29, The University;
  1998, 2003.

\bibitem{Sudeep}
S.~Kamath, D.~N.~C. Tse, and C.-C. Wang, ``Two-unicast is hard,'' in {\em
  {IEEE} International Symposium on Information Theory}, pp.~1--5, 2014.

\bibitem{BoydCVX}
S.~Boyd, {\em Convex Optimization}.
\newblock Cambridge University Press, 2004.

\end{thebibliography}
